\documentclass[12pt,a4paper]{article}

\usepackage{amsmath,amsthm,amssymb,amscd,a4wide}

\usepackage{dsfont}
\usepackage{mathrsfs}
\usepackage{setspace}
\usepackage{hyperref}
\usepackage{color}
\usepackage{enumerate}

\allowdisplaybreaks

\newcommand{\Nlargest}{\mathcal{N}_{N}^{>}}

\newcommand{\WNj}{W_N^j}
\newcommand{\LN}{L_N}
\newcommand{\lj}{l_j}
\newcommand{\llargest}{\ell_{N}^{>}}
\newcommand{\nuN}{\nu_N}

\newcommand{\gN}{g_N}

\newcommand{\muN}{\mu_N}

\newcommand{\lambdaN}{\lambda_N}
\newcommand{\xiN}{\xi_N}

\newcommand{\phiGP}[2]{\phi^{\mathrm{GP}}_{#1,#2}}
\newcommand{\Mlargest}{\mathcal M_N^>}
\newcommand{\nMjgz}{n_{M_j}}
\newcommand{\llargesttilde}{\tilde{\ell}_N^{>}}
\newcommand{\E}{\mathds E}

\newcommand{\ud}{\mathrm{d}}

\DeclareMathOperator{\tr}{Tr}

\newcommand{\eins}{\mathds{1}}

\numberwithin{equation}{section}

\newtheorem{theorem}{Theorem}[section]
\newtheorem{lemma}[theorem]{Lemma}

\newtheorem{cor}[theorem]{Corollary}
\newtheorem{remark}[theorem]{Remark}

\theoremstyle{definition}

\numberwithin{equation}{section}

\begin{document}
	
	\thispagestyle{empty}
	
	\vspace*{1cm}
	
	\begin{center}
		
		{\Large \bf Bose--Einstein condensation in the Luttinger--Sy model with contact interaction}\\ 

		\vspace*{2cm}
		
		{\large Joachim Kerner \footnote{E-mail address: {\tt joachim.kerner@fernuni-hagen.de}}, Maximilian Pechmann \footnote{E-mail address: {\tt maximilian.pechmann@fernuni-hagen.de}}, and Wolfgang Spitzer \footnote{E-mail address: {\tt wolfgang.spitzer@fernuni-hagen.de}}}%
		
		\vspace*{5mm}
		
		Department of Mathematics and Computer Science\\
		FernUniversit\"{a}t in Hagen\\
		58084 Hagen\\
		Germany\\
		
	\end{center}
	
	\vfill
	
	\begin{abstract}
		We study bosons on the real line in a Poisson random potential (Luttinger--Sy model) with contact interaction in the thermodynamic limit at absolute zero temperature. We prove that generalized Bose--Einstein condensation (BEC) occurs almost surely if the intensity $\nu_N$ of the Poisson potential satisfies $[\ln (N)]^4/N^{1 - 2\eta} \ll \nuN \lesssim 1$ for arbitrary $0 < \eta \leq 1/3$. We also show that the contact interaction alters the type of condensation, going from a type-I BEC to a type-III BEC as the strength of this interaction is increased.  Furthermore, for sufficiently strong contact interactions and $0 < \eta < 1/6$ we prove that the mean particle density in the largest interval is almost surely bounded asymptotically by $\nu_NN^{3/5+\delta}$ for $\delta > 0$.
	\end{abstract}
	
	\newpage
	
%
	\section{Introduction}
	Bose--Einstein condensation (BEC) generally refers to a macroscopic occupation of a single-particle state. In the non-interacting Bose gas and as illustrated by the initiating work of Einstein~\cite{PaperEinstein1,PaperEinstein2}, this one-particle state is the ground state of the one-particle Hamiltonian while, in the more general case of an interacting gas, one considers the ground state of the reduced one-particle density matrix instead \cite{PO56}. More general definitions of BEC leading to the notion of generalized BEC require a macroscopic occupation of an arbitrarily narrow energy band of single-particle states \cite{casimir1968bose, van1982generalized, van1983condensation, van1986general, van1986general2,ZagBru} and condensation in this generalized sense is thought to be thermodynamically more stable \cite{girardeau1960relationship, jaeck2010nature}. More explicitly, BEC is classified into three different types: Type-I or type-II is present whenever finitely or infinitely many single-particle states in this narrow energy band are macroscopically occupied. A generalized BEC without any single-particle state in this band being macroscopically occupied is defined as type-III. 
	
	In this paper, we are concerned with (generalized) BEC in the Luttinger--Sy model \cite{LuttSyEnergy,luttinger1973bose} with a contact interaction of strength $g$ of the Lieb--Liniger type~\cite{LL63} in the thermodynamic limit at absolute zero temperature. In particular, we are interested in determining its type and in estimating the maximal particle density per interval. Note that the Luttinger--Sy model is a random model with the real line as one-particle configuration space in which singular point impurities (or singular external potentials for that matter) are Poisson distributed and consequently divide the real line into a countable number of disjoint intervals. In the original model, the intensity $\nu$ of the Poisson distribution was kept fixed but in this paper we will allow it to vary with the number of particles, $N$. Likewise, the strength $g$ of the contact interaction depends on $N$. In fact, we always assume that $g_N$ goes to $0$ as $N$ becomes large. 
	
	It is interesting to note that, despite its singularity, the Luttinger--Sy model is considered a good approximation to more general Poisson random potentials \cite{zagrebnov2007bose}. Regarding Bose--Einstein condensation, being in accordance with a conjecture for bosonic systems with quite general random potentials \cite{lenoble2004bose}, BEC in the non-interacting Luttinger--Sy model is of type-I and leads, in the thermodynamic limit, to an unbounded particle density of order $\ln (N)$ in the largest interval \cite{zagrebnov2007bose}. Due to the diverging particle density, interactions between bosons cannot be neglected and BEC has to be investigated for the interacting Luttinger--Sy model. So far, however, only a limited amount of rigorous results regarding BEC in (random and non-random) interacting bosonic systems exist, e.g. \cite{KennedyLiebShastry,LiebSeiringerProof,LVZ03,Lieetal05}.
	This is even more true for one-dimensional many-particle systems \cite{SeiYngZag12,PS86} or quasi one-dimensional systems such as quantum graphs \cite{BolteKernerBEC,BolteKernerInstability} since Bose-Einstein condensation in one dimension is much more unstable \cite{Hohenberg,LanWil79}. 
	
	In large parts we follow the paper \cite{SeiYngZag12} by Seiringer, Yngvason, and Zagrebnov, in which an equivalent model is considered, however, with the unit interval as the fixed one-particle configuration space. This implies that the thermodynamic limit is a high-density limit. There, BEC is exclusively discussed in the sense of a macroscopic occupation of a single-particle state. These authors were able to prove condensation and to make conclusions of its localization among the intervals, assuming a fast decay of the intensity, i.e., $1 / N  \ll\nuN \ll [\ln (N)]^3 / N$, cf. Appendix \ref{Appendix Seiringer gN nuN limit}. In the present paper, by considering BEC in the generalized sense, we are able to prove condensation in the almost sure sense and to obtain knowledge of the distribution of the condensate among the intervals for an intensity $[\ln(N)]^4/N^{1-\eta} \ll \nuN \lesssim 1$ for any $0 < \eta \leq 1/3$, see Theorem \ref{Theorem generalized BEC}. In addition, we show that the type of the BEC changes with the strength of the contact interaction, see Theorem \ref{Theorem transition of condensation}.

We also provide an upper bound for the particle density in the largest interval that depends on the strength of the interaction. In particular, we prove that for strong interactions $g_N \gg \nu_N N^{-1/6} [\ln(N)]^{-2}$, this density is almost surely non-divergent in case of $[\ln(N)]^4/N^{1-2\eta} \ll \nuN \lesssim N^{-3/5 - \delta}$ for any $0 < \eta < 1/6$, $0 < \delta < 2/5 - 2\eta$ and is asymptotically bounded by $\nu_N N^{3/5 + \delta}$ for any $\delta> 0$ in the case of $[\ln(N)]^4/N^{1-2\eta} \ll \nuN \lesssim 1$, $0 < \eta < 1/6$, see Theorem \ref{TheoremParticleDensity}.

	%

%
\section{The model}\label{TheModel}
We consider bosons on the real line that interact pairwise repulsively via a delta-function potential as in the Lieb--Liniger model~\cite{LL63} and with an external random potential as in the Luttinger--Sy model~\cite{luttinger1973bose}. For given intensity $\nu>0$ we think of a Poisson point process $X$ on $\mathds{R}$ as a random variable on some probability space $(\Omega,\mathcal F,\mathds{P})$. We assume that $X(\omega)=\{x_j(\omega): j\in \mathbb{Z}\}$ is a strictly increasing sequence of points $x_j = x_j(\omega)\in \mathds{R}$ and that $0\in(x_0(\omega), x_1(\omega)$. For any bounded Borel set $\Lambda\subset \mathds{R}$ with Lebesgue volume $|\Lambda|$, the probability that $\Lambda$ contains exactly $m$ points $x_j(\omega)$ is
\begin{equation}
\mathds{P}( \mathrm{card}(X(\omega)\cap\Lambda) = m) = \dfrac{(\nu |\Lambda|)^m}{m!}\, \mathrm{e}^{-\nu |\Lambda|}\,,\quad m\in \mathds{N}_0\ .
\end{equation}
And, if $\Lambda$ and $\Lambda'$ are two such subsets which are disjoint, then the events $\{X(\omega)\cap\Lambda\}$ and $\{X(\omega)\cap\Lambda'\}$ are (stochastically) independent.

In this paper, we will write $\mathds P = \mathds P^{\nu}$ and $\E = \E^{\nu}$ if $\nu > 0$ is constant. The external random potential we consider is then of the form $\sigma\sum_{j\in\mathbb{Z}} \delta(z-x_j(\omega))$. Actually, we restrict the analysis to the case of infinite strength, that is, we set informally $\sigma=\infty$. The pair interaction is also described by a delta-function but now of finite strength $g\geq0$; in fact, $g$ will tend to zero as the number of particles increases. In order to define the full many-boson Hamiltonian we only need (since $\sigma=\infty$ implies Dirichlet boundary conditions at the endpoints of the intervals) to define the many-boson Hamiltonian on a bounded interval and take their direct sum.  

When we perform the thermodynamic limit we define for any particle density $\rho>0$ and any particle number $N\in\mathds{N}$ the length $L_N := N/\rho$ and introduce the interval $\Lambda_N := (-L_N/2,L_N/2)$. Later, we will define $W_j = W_j(\omega) := (x_j(\omega),x_{j+1}(\omega))\cap \Lambda_N$ for any realization $X(\omega)$ and denote the number of non-empty intervals $W_j$ by $k_N$, which almost surely is finite and has the asymptotic behavior $\lim_{N \to \infty} k_N/L_N = \nu$ \cite{zagrebnov2007bose}, see also Theorem~\ref{Theorem k LN to nu}. But for now, we consider an arbitrary interval $\Lambda \subset \mathds R$, an arbitrary set $\{x_j : j \in \mathbb{Z}\} \subset \mathds R$ such that only finitely many are contained in $\Lambda$ and define $W_j := (x_j,x_{j+1})\cap \Lambda$ . Let $l_j := |W_j|$ be the (later random) length of the interval $W_j$. Note that $\sum_{j\in\mathbb{Z}} l_j= L_N$ and $\text L^2(\Lambda) = \bigoplus_{j\in\mathbb{Z}} \text L^2(W_j)$. The inner product is always denoted by $\langle\cdot,\cdot\rangle$ with norm $\|\psi\| := \langle\psi,\psi\rangle^{1/2} := ( \int |\psi(x) |^2\, \text d x)^{1/2}$. 

We impose Dirichlet boundary conditions for the Laplacian $-\partial^2_{{z}}$ at the end points of an interval $W$ with length $l:= |W|>0$. For any $n\in\mathds{N},g\ge0$ we define the $n$-boson Hamiltonian on the interval $W$,
\begin{equation}\label{HamiltonianLiebLiniger}
H(n,l,g) :=-\sum\limits_{i=1}^{n} \partial^2_{{z_i}} + g \sum\limits_{1\le i < j\le n}
\delta(z_i - z_j)\ , \quad g \geq 0\ ,
\end{equation}
acting as a quadratic form on the $n$-fold symmetric tensor product $\bigotimes_s^n\text L^2(W)$. For completeness, we set $H(0,l,g)=0$ for any $l\geq0$.

 Now, let $\Lambda,W_j$, and $l_j$ be as above. We call $\boldsymbol N= \{M_j:j\in\mathbb{Z}\}$ an admissible sequence (ad.~seq.) of the particle number $N$ if $M_j\in\mathds{N}_0$, if $M_j = 0$ for any $j \in \mathbb{Z}$ with $l_j = 0$, and if the total number of particles is $\sum_{j\in\mathbb{Z}} M_j = N$. Furthermore, we call $\{ M_j \}_{j \in \mathbb{Z}}$ a general admissible sequence (gen.~ad.~seq.) if it fulfills the first two requirements of but not the last. Then the full $N$-boson Hamiltonian on the interval $\Lambda$,
\begin{equation}\label{AdmissibleHamiltonian} H(\boldsymbol N,\Lambda,g) := \bigoplus_{j\in\mathbb{Z}} H(M_j,l_j,g) \,,
\end{equation}
acts on the $N$-fold symmetric tensor product $\bigotimes_s^N\text L^2(\Lambda)$. 

For any (general) admissible sequence $\{M_j\}_{j \in \mathbb{Z}}$ we denote by $\mathcal M_N^>:=\max \{M_j : j\in\mathbb{Z} \}$ the largest particle number $M_j$. Sometimes we will also need the (later random) lengths of the intervals $l_j$ arranged in descending order. We denote them by $\llargest =: l_N^{>,1} \ge l_N^{>,2} \ge \ldots$, $\llargest = \max\{l_j:j\in\mathbb{Z}\}$ being the length of the largest subinterval $W_j$ of $\Lambda$. Note here that we added the particle number $N$ as an index since these numbers will eventually depend on $N$.


One of our main concerns is the ground-state energy of the full $N$-boson Hamiltonian in the thermodynamic limit. So, let $E_0^{\mathrm{QM}}(n,l,g)$ and $E_0^{\mathrm{QM}}(\boldsymbol N,\Lambda,g)$ be the ground-state energies of the Hamiltonians $H(n,l,g)$ and $H(\boldsymbol N,\Lambda,g)$, respectively. That is,
\begin{align}
E_0^{\mathrm{QM}}(n,l,g) &:= \inf \Big\{\langle \psi, H(n,l,g)\psi\rangle: \psi\in\bigotimes_s^n\text H^1_0(W), \|\psi\| =1\Big\}\,,
\\
E_0^{\mathrm{QM}}(\boldsymbol N,\Lambda,g) &:= \inf \Big\{\langle \psi, H(\boldsymbol N,\Lambda,g)\psi\rangle: \psi\in \bigoplus_{j} \bigotimes_s^{M_j}\text H^1_{0}(W_j) ,\|\psi\| =1\Big\}\ .
\end{align} 
The latter energy will eventually be random due to the random location of the points $x_j$ that partition the interval $\Lambda_N$ into the intervals $W_j$. 

The relevant quantity is the lowest ground-state energy among all possible 
distributions of particles in the intervals $W_j$, i.e., 
\begin{equation}\begin{split}
E_{\mathrm{LS}}^{\mathrm{QM}}(N,\Lambda,g) &:= \inf 
\Big\{E_0^{\mathrm{QM}}(\boldsymbol N,\Lambda,g):\boldsymbol N \text{ ad.~seq. of 
} N\Big\}
\\
& = \inf \Big\{\sum_{j\in\mathbb{Z}} E_0^{\mathrm{QM}}(M_j,l_j,g): \{ M_j\}_{j \in \mathbb{Z}} \text{ ad.~seq.}, \sum_{j\in\mathbb{Z}} {M_j} = N\Big\} \\
& = \inf \Big\{\sum_{j\in\mathbb{Z}} E_0^{\mathrm{QM}}(M_j,l_j,g): \{ M_j\}_{j \in 
\mathbb{Z}} \text{ gen.~ad.~seq.}, \sum_{j\in\mathbb{Z}} {M_j} \ge N\Big\} 
\label{E0QM Mj ge N}\ .
\end{split}
\end{equation}
Note that, for fixed $\Lambda \subset \mathds R$ and $g \ge 0$, $E_{\mathrm{LS}}^{\mathrm{QM}}(N,\Lambda,g)$ is non-decreasing in $N$. 
\begin{remark} The energy $E_{\mathrm{LS}}^{\mathrm{QM}}(N,\Lambda,g)$ is (for $\sigma = \infty$) the ground-state energy of the self-adjoint operator
	\begin{align}
	H(N,\Lambda,g) := - \sum\limits_{i=1}^N \partial_{z_i}^2 + g_N \sum\limits_{1 
		\le i < k \le N} \delta(z_i - z_k) + \sigma \sum\limits_{j \in \mathds Z} 
	\delta(z - x_j) 
	\end{align}
	defined on $\bigotimes\limits_s^N L^2(\Lambda)$.
\end{remark}
The ground-state energy $E_0^{\mathrm{QM}}(\boldsymbol N,\Lambda,g)$ will be approximated by the ground-state energy of a mean-field Hamiltonian $h$: In order to introduce this operator we fix some interval $W$ of length $l = |W|>0$ and, as above, let $g\geq0$. For $n > 0$, the Gross--Pitaevskii (GP) functional $\mathcal E^{\mathrm{GP}}(n,l,g)$ on $W$ with domain $\{ \phi \in \text H_0^1(W) : \int_W|\phi(z)|^2 \, \mathrm{d} z = n\}$ then maps a function $\phi$ from the Sobolev space $\text H_0^1(W)$ to 
\begin{equation}
\mathcal E^{\mathrm{GP}}(n,l,g)[\phi] := \int\limits_{W} \left(|\phi^{\prime}(z)|^2 + \dfrac{g}{2} |\phi(z)|^4 \right)\ud z\,.
\end{equation}
%
%
As is well-known \cite{lieb2000bosons}, there is a unique, non-negative minimizer of $\mathcal E^{\mathrm{GP}}(n,l,g)$, which we denote by $\phi^{\mathrm{GP}}_{n,l,g} \in \text H_0^1(W)$. Let
\begin{equation} E^{\mathrm{GP}}(n,l,g) := \mathcal E^{\mathrm{GP}}(n,l,g)[\phi^{\mathrm{GP}}_{n,l,g}] = \inf\Big\{\mathcal E^{\mathrm{GP}}(n,l,g)[\phi]:\phi\in \text H_0^1(W),\|\phi\|^{2} = n\Big\}\,.
\end{equation} 
Setting $e^{\mathrm{GP}}(g) := E^{\mathrm{GP}}(1,1,g)$, we obtain by scaling \cite{lieb2004one}
\begin{equation}\label{EGP scaling}
E^{\text{GP}}(n,l,g) = n E^{\mathrm{GP}}(1,l,ng) = \dfrac{1}{l^2} E^{\mathrm{GP}}(n,1,lg)  = \dfrac{n}{l^2} e^{\mathrm{GP}}(nlg)\ .
\end{equation}
We also set $E^{\mathrm{GP}}(0,l,g) := 0$ and $\phi^{\mathrm{GP}}_{l,g} := \phi^{\mathrm{GP}}_{1,l,g}$ for any $l, g \ge 0$. 

Finally, the self-adjoint (one-particle) mean-field Hamiltonian $h=h(l,g)$ on the one-particle Hilbert space $\text L^2(W)$ shall be 
\begin{equation}\label{mean-field Hamiltonian 1}
h(l,g) := - \partial_{z}^2 + g | \phi^{\mathrm{GP}}_{l,g}|^2 - \dfrac{g}{2}  \int\limits_{W} | \phi^{\mathrm{GP}}_{l,g}(z)|^4\ \ud z\ ,
\end{equation}
and on $\bigotimes_s^n\text L^2(W)$ we introduce
\begin{equation}\label{mean-field Hamiltonian 2}
h^{(i)}(l,g) := \eins \otimes \cdots \otimes \eins \otimes h(l,g)\otimes\eins \otimes\cdots\otimes\eins \,,
\end{equation}
where $h(l,g)$ acts on the $i$th component in the $n$-fold tensor product. 
\begin{remark}\label{RemarkGroundStateMeanFieldHamiltonian} The minimizer $\phi^{\mathrm{GP}}_{l,g}$ is also the ground state of $h(l,g)$ with corresponding ground-state energy $E^{\mathrm{GP}}(1,l,g)$.
\end{remark}
%
%
Now, in a first result we compare the operator $H(n,l,g)$ with the second quantization of $h(l,g)$, i.e., $\sum_{i}h^{(i)}(l,g)$.
\begin{lemma} \label{lemma 1} Let $n\in \mathds{N}_0$, $l>0$, and $g > 0$ be given. Then there exist finite, positive constants $c,\tilde{c}$ (independent of $n$ and $l$) such that if $\tilde{c}(n^{1/3}lg)^{1/2} < 1$ we have, for some $\tau$ with $0< \tau < \frac{\tilde{c}}{2}(n^{1/3}lg)^{1/2}$,
	\begin{align} \label{equation H lower bound}
	H(n,l,g) & \ge (1 - \tau) \sum\limits_{i=1}^{n} h^{(i)}(l,ng) + \tau n l^{-2}e^{\mathrm{GP}}(nlg)  - cn^{5/3} gl^{-1} \left( e^{\mathrm{GP}}(nlg) \right)^{1/2}\,.
	\end{align} 
\end{lemma}
This estimate is contained in the proof of \cite[Theorem~2.1]{SeiYngZag12} but we take a slightly different route, which is the reason why we recall the main steps. 
\begin{proof} The statement is trivial for $n = 0$ and we therefore assume $n >0$ in the following. Let $\epsilon \in (0,1)$ and $b > 0$ be given. As demonstrated in \cite{SeiYngZag12}, one has the operator inequality 
	\begin{equation}
	-\epsilon \partial^2_z + g \delta(z) \ge \dfrac{g}{1+bg/(2\epsilon)} \delta_b(z)\,,
	\end{equation}
	where $\delta_b(z) := \frac1{2b} \exp(-|z|/b)$ is a function of positive type. 
	
	Now, setting $p_{i}:=-\mathrm{i}\partial_{z_i}$ we follow \cite[(7--15)]{SeiYngZag12} to obtain
	\begin{equation}\begin{split}
	H(n, l, g) &\ge \sum\limits_{i=1}^{n} \left[ \left(1 - \dfrac{\epsilon}{2} \right) p_{i}^2 \right] + g \left( 1 + \dfrac{b n g}{2 \epsilon} \right)^{-1} \sum\limits_{1\leq i < k\leq N} \delta_b(z_i - z_k) \\
	&\ge \sum\limits_{i=1}^{n} \left[ \left(1 - \dfrac{\epsilon}{2} \right) p_{i}^2  + n g | \phi^{\mathrm{GP}}_{l,ng} (z_i) |^2 - \dfrac{n g}{2} \int\limits_{W} | \phi^{\mathrm{GP}}_{l,ng} (z) |^4\ \ud z \right] \, + \\
	&\ + \, \sum\limits_{i=1}^{n} \left[ - \dfrac{b (n g)^2}{2 \epsilon} \left( E^{\mathrm{GP}}(1,l,ng) \right)^{1/2} - c n g \left( E^{\mathrm{GP}}(1,l,ng) \right)^{3/4} b^{1/2} - \dfrac{n g}{4 n b} \right]
	\end{split}
	\end{equation}
	for some constant $c> 0$.
	Since
	\begin{equation}\begin{split}
	&\sum\limits_{i=1}^{n} \left[ \left(1 - \dfrac{\epsilon}{2} \right) p_{i}^2  + n g | \phi^{\mathrm{GP}}_{l,ng} (z_i) |^2 - \dfrac{n g}{2} \int\limits_{W} | \phi^{\mathrm{GP}}_{l,ng}(z) |^4\ \ud z \right] \\
	&\quad \ge \sum\limits_{i=1}^{n}  \left[ \left(1 - \dfrac{\epsilon}{2} \right) h^{(i)}(l,ng) - \dfrac{\epsilon}{2} E^{\mathrm{GP}}(1,l,ng) \right] \ , 
	\end{split}	
	\end{equation}
	we obtain
	\begin{equation}
	\begin{split}
	H(n, l, g) &\ge \sum\limits_{i=1}^{n}  \left[ \left(1 - \dfrac{\epsilon}{2} \right) h^{(i)}(l,ng)  +\dfrac{\epsilon}{2} E^{\mathrm{GP}}(1,l,ng)  \right] \\		
	& \quad - \, \sum\limits_{i=1}^{n} \left[ \epsilon E^{\mathrm{GP}}(1,l,ng)+\dfrac{b (n g)^2}{2 \epsilon} \left( E^{\mathrm{GP}}(1,l,ng) \right)^{1/2} \right. \\
	&\quad \quad  + \left. cng \left(E^{\mathrm{GP}}(1,l,ng) \right)^{3/4} b^{1/2} + \dfrac{n g}{4 n b}\right] \ .
	\end{split}
	\end{equation}
Next, choosing $\epsilon= 2^{-1/2} b^{1/2} n g  \left( E^{\mathrm{GP}}(1,l,ng) \right)^{-1/4}$ and $b= \tilde{c}^2 n^{-2/3} \left( E^{\mathrm{GP}}(1,l,ng) \right)^{-1/2}$ with some constant $\tilde{c} > 0$ yields, $\tau:= \epsilon/2$, 
	\begin{equation}
	H(n, l, g) \ge \left( 1 - \tau \right) \sum\limits_{i=1}^{n} h^{(i)}(l,ng) + \tau n E^{\mathrm{GP}}(1,l,ng) - \hat{c} n  \left( n \right)^{2/3} g \left( E^{\mathrm{GP}}(1,l,ng) \right)^{1/2}
	\end{equation}
	with some constant $\hat{c} > 0$.
	
Finally, note that the bound on $\tau$ follows from the definitions of $\epsilon$ and $b$ and from 
	\begin{equation} \label{bound EGP}
	E^{\mathrm{GP}}(1,l,ng) = \dfrac{1}{l^2} e^{\mathrm{GP}}(nlg) \ge \dfrac{nlg}{2 l^2} \int\limits_0^1 | \phiGP{1}{nlg} (z)|^4\ \ud z \ge \dfrac{n g}{2 l} \ ,
	\end{equation}
	which holds due to 	
	\begin{equation}
	\int\limits_0^1| \phiGP{1}{nlg} (z)|^4\ \ud z = \left( \Big\| | \phiGP{1}{nlg} |^2  \Big\|_{\text L^2(0,1)} \| 1 \|_{L^2(0,1)} \right)^2 \ge \left( \Big\| | \phiGP{1}{nlg} |^2 \cdot 1 \Big\|_{\text L^1(0,1)} \right)^2 = 1\ ,
	\end{equation}
	where we used H\"older's inequality. 
\end{proof}
In a next step we bound the ground-state energy of each component in~\eqref{AdmissibleHamiltonian} in terms of the corresponding Gross--Pitaevskii energy.
\begin{theorem}[An energy bound]\label{thm energy bound} Let $0 < \epsilon < 1$, $g \geq 0$ and $N \in \mathds N$ be given. Then, for any general admissible sequence $\{ M_j \}_{j\in\mathbb{Z}}$ with associated lengths $\{l_j\}_{j\in \mathbb{Z}}$ that fulfills inequality $[ (\Mlargest)^{1/3} \llargest g]^{1/2} < \epsilon / \max \{ \sqrt{2} c , \tilde{c} \}$ (where $c$ and $\tilde{c}$ are the constants from Lemma~\ref{lemma 1}), one has
	\begin{equation}\label{energy bound}
	M_j E^{\mathrm{GP}}(1, l_j, M_j g) \ge E^{\mathrm{QM}}_0(M_j,l_j,g) \ge \left( 1 - \epsilon \right)M_j E^{\mathrm{GP}}(1, l_j, M_j g) \ ,
	\end{equation}
	for all $j \in \mathbb{Z}$ where $l_j > 0$.
\end{theorem}
\begin{proof}
	Since the inequality \eqref{energy bound} is trivial for $M_j = 0$, we assume $M_j \ge 1$ in the following.
	
	Upper bound: this follows directly from a standard variational argument using the product state $\otimes \phi^{\mathrm{GP}}_{\lj,M_j g}$. 
	
	Lower bound: Since $h^{(i)}(\lj,M_j g) \ge E^{\mathrm{GP}}(1, \lj, M_j g)$ for $j \in \mathbb{Z}$ with $M_j \ge 1$ by Remark~\ref{RemarkGroundStateMeanFieldHamiltonian} and since
	\begin{equation}
	\tilde{c} \left[ M_j^{1/3} \lj g \right]^{1/2} \le \tilde{c} \left[\left( \Mlargest \right)^{1/3} \llargest g \right]^{1/2} < 1 \, ,
	\end{equation}
	\eqref{equation H lower bound} implies
	\begin{equation} \label{equation EGP not bounded}
	E_0^{\mathrm{QM}}(M_j,l_j,g) \ge M_j E^{\mathrm{GP}}(1, l_j, M_j g) \left( 1 - c M_j^{2/3} g \left( E^{\mathrm{GP}}(1, l_j, M_j g) \right)^{-1/2} \right) \ .
	\end{equation}
	Finally, applying inequality \eqref{bound EGP} then yields the statement.
\end{proof}
In a next result we estimate the fraction of particles occupying the Gross--Pitaevskii ground state $\phi^{\mathrm{GP}}_{l_j,M_j g}$. For this note that
\begin{align}\label{definition occupation number}
n_{M_j}:=\tr_{\mathrm{L}^2(\WNj)}[\rho^{(1)}_{M_j} \phi^{\mathrm{GP}}_{l_j,M_j g}(\phi^{\mathrm{GP}}_{l_j,M_j g},\cdot)]
\end{align}
is the number of particles occupying $\phi^{\mathrm{GP}}_{l_j,M_j g}$. Here
\begin{align}
\rho^{(1)}_{M_j} := \begin{cases}
M_j \tr_{\mathcal H_{W_j}^{(M_j-1)}}[\rho_{M_j}]  & \quad \text{ if } M_j \ge 2\ ,\\
\rho_1 & \quad \text{ if } M_j = 1\ ,
\end{cases} 
\end{align}
with $\mathcal H_{W_j}^{(M_j-1)} := S_{M_j - 1} \text L^2(W_j^{M_j - 1})$ and $S_{M_j - 1}$ being the symmetrizer on $\text L^2(W_j^{M_j - 1})$, is the reduced one-particle density matrix that is obtained from the many-particle ground state $\rho_{M_j}$ of $H(M_j, l_j, g)$ by taking the partial trace. Most importantly, $\tr \big[ \sum_{i=1}^{M_j} h^{(i)}(l_j,g) \rho_{M_j}\big] = \tr\big[h(l_j,g) \rho^{(1)}_{M_j}\big]$, see~\cite{M07} for more details. 
\begin{theorem}[Occupation number of single-particle state]\label{Theorem 2}
	Let $N \in \mathds N$ and a general admissible sequence $\{M_j \}_{j\in\mathbb{Z}}$ with associated lengths $\{l_j\}_{j\in \mathbb{Z}}$. Suppose $\tilde c\left[ ( \Mlargest )^{1/3} \llargest g \right]^{1/2}< 1/2$ (with the constant $\tilde c$ from Lemma~\ref{lemma 1}). Then for any $j \in \mathds Z$ with $M_j \ge 1$,
	\begin{equation}\label{BEC inequality}
	\left( 1 - \dfrac{\nMjgz}{M_j} \right) \le \sqrt{7} c \cdot \left\{ \ln \left[ 1 + \mathrm{e}^{-2 (\pi^2 + 3 M_j l_j g)^{1/2}} \right] \right\}^{-1}  M_j^{2/3} l_j g 
	\end{equation}
	with the constant $c$ from Lemma~\ref{lemma 1}.
\end{theorem}
\begin{proof}
	We write $E^{\mathrm{QM}}_1(1,\lj,M_j g)$ for the second eigenvalue of the mean-field Hamiltonian $h(1,\lj, M_j g)$. After tracing \eqref{equation H lower bound} with the density matrix associated to the ground state of the Hamiltonian $H(M_j,\lj,g)$ as in \cite[proofs of Theorems 2.1 and 2.2]{SeiYngZag12}, one obtains
	\begin{equation}
	\begin{split}
	E^{\mathrm{QM}}_0(M_j, \lj, g) 
	& \ge (1-\tau) \Big[n_{M_j} E^{\mathrm{GP}}(1,\lj ,M_j g) + (M_j - 
n_{M_j})  E^{\mathrm{QM}}_1(1,\lj,M_j g) \Big] \, + \\
	& \quad + \, \tau M_j \lj^{-2} e^{\mathrm{GP}}(M_j \lj g) - c M_j^{5/3} 
g \lj^{-1} \left( e^{\mathrm{GP}}(M_j \lj g) \right)^{1/2} \ .
	\end{split}
	\end{equation}
	Note here that inequality~\eqref{BEC inequality} is slightly improved if compared to \cite[Theorem 2.2]{SeiYngZag12} because we omit an estimation of $(e^{\mathrm{GP}}(M_j \lj g))^{1/2}$ in the last term of this inequality. Then, using the upper bound of Theorem~\ref{thm energy bound} we obtain
	\begin{align}
	\left( 1 - \dfrac{\nMjgz}{M_j} \right) & \le c(1 - \tau)^{-1} \dfrac{\left( E^{\mathrm{GP}}(1,\lj ,M_j g) \right)^{1/2}}{E^{\mathrm{QM}}_1(1,\lj ,M_j g) - E^{\mathrm{GP}}(1,\lj ,M_j g)} M_j^{2/3} g \\
	& \le 2 c \dfrac{ \left( e^{\mathrm{GP}}(M_j \lj g) \right)^{1/2}}{ \left( \pi^2 + 3 M_j \lj g \right)^{1/2}} \left\{ \ln\left( 1 + \pi \mathrm{e}^{-2 \sqrt{\pi^2 + 3 M_j \lj g}} \right) \right\}^{-1} M_j^{2/3} \lj g
	\end{align}
	for all sufficiently large $N \in \mathds N$. For the last estimate we used
	\begin{align}
	E^{\mathrm{QM}}_1(1,l , n g) - E^{\mathrm{GP}}(1,l , n g) \ge \dfrac{\eta}{l^2} \ln \left( 1 + \pi \mathrm{e}^{-2 \eta} \right)
	\end{align}
	with $\eta = \sqrt{\pi^2 + 3 n l g}$ and $l,g> 0$, see \cite{kirsch1985universal} and \cite[(A.8)]{SeiYngZag12}. 
	
	The statement then follows with
	\begin{align}
	\dfrac{e^{\mathrm{GP}}(M_j \lj g)}{\pi^2 + M_j \lj g} \le \dfrac{e^{\mathrm{GP}}(M_j \lj g) - e^{\mathrm{GP}}(0)}{M_j \lj g} + \dfrac{\pi^2}{\pi^2} \le \dfrac{7}{4} \ ,
	\end{align}
	where we used $(e^{\mathrm{GP}}(\kappa) - e^{\mathrm{GP}}(0))/ \kappa \le 3/4$ \cite[(31)]{SeiYngZag12} for $\kappa > 0$ as well as $e^{\mathrm{GP}}(0)= \pi^2$.
\end{proof}
From now on we assume a partition of $\Lambda_N = (- L_N/2, L_N/2)$ caused by a Poisson point process as described at the beginning. Moreover, we now allow the intensity of impurities $\nu$ and the pair interaction $g$ to vary with the particle number $N$. In order to account for variable impurity intensities, we do the following: We introduce a sequence $(s_N)_{N \in \mathds{N}} \subset \mathds{R}$ such that $s_N \leq 1$ and perform the scaling $x_j(\omega) \mapsto s_N^{-1} x_j(\omega)$ and we set $\nu_N := \nu s_N$. If we consider a constant intensity $\nu = const.$ we set $s_N = 1$ for any $N \in \mathds N$.
We define
\begin{equation}
\tilde l_j := |\widetilde W_j| := \left| \left( s_N^{-1} x_j(\omega), s_N^{-1} x_{j+1}(\omega) \right) \cap \Lambda_N \right| \ .
\end{equation}
We have $\tilde{l}_j = s_N^{-1} l_j$ for any $\tilde{l}_j > 0$ except possibly for the first and last subinterval $\widetilde W_j$ within $\Lambda_N$. We write $\tilde k_N$ for the number of the scaled subintervals within the window $\Lambda_N$. Then $\lim_{N \to \infty} \tilde k_N / (L_N \nu_N) = 1$ almost surely by scaling. We also define $\llargesttilde := \max\{ \tilde l_j : j \in \mathds Z \}$.
\begin{remark}
	Comparing Theorem~\ref{thm energy bound} and \ref{Theorem 2} to \cite[Theorems~2.1 and~2.2]{SeiYngZag12} one observes that we only require $(\Mlargest)^{1/3} \llargesttilde \gN$ instead of $N^{1/3} \LN \gN$ to converge to zero. 
	This will allow us to consider stronger interactions in the sense that $\gN$ converges more slowly to zero.
\end{remark}
Note that the lengths $|(x_j(\omega), x_{j+1}(\omega)|$ are exponentially distributed random variables with parameter $\nu$ for any $j \in \mathds Z \backslash \{ 0\}$ \cite[Ch. 4]{kingman1993poisson}. For any $N \in \mathds N$ and any $l, \widetilde \mu > 0$ we define the function $l \mapsto N_{\gN,\widetilde \mu}(l)$ to be the unique minimizer of the Legendre transformation of the map $\widetilde N \mapsto E^{\mathrm{GP}}(\widetilde N, l, \gN)$ (see \cite{SeiYngZag12} for details). More explicitly, $N_{\gN,\widetilde \mu}(l)$ is such that
\begin{equation}\label{determinging formula for Ngmu}
E^{\mathrm{GP}}(N_{\gN,\widetilde \mu}(l), l, \gN) - \widetilde \mu N_{\gN, \widetilde \mu}(l)= \inf\limits_{\widetilde N \geq 0} \left( E^{\mathrm{GP}}(\widetilde N, l, \gN)) - \widetilde \mu \widetilde N \right) \ .
\end{equation}
We remark that  $N_{\gN,\widetilde \mu}(l)$ obeys
\begin{equation} \label{inequality for Ngmu(l)}
\dfrac{2}{3} \dfrac{1}{l \gN} \left[ \widetilde \mu l^2 - \pi^2 \right]_+ \le N_{\gN,\widetilde \mu}(l) \le \dfrac{1}{l \gN} \left[ \widetilde \mu l^2 - \pi^2 \right]_+
\end{equation}
%
%
with $[x]_+ := \max\{x,0\}$. Consequently,
\begin{align}
\frac{2}{3} \frac{\widetilde \mu}{g_N \nu} \xi \mathrm{e}^{- \pi \nu / \sqrt{\widetilde \mu}} \le \mathds E \left[ N_{g_N,\widetilde \mu} \right] = \int\limits_0^{\infty} N_{g_N,\widetilde \mu}(l) \, \nu \mathrm{e}^{- \nu l} \mathrm{d} l  \le  \frac{\widetilde \mu}{g_N \nu} \xi \mathrm{e}^{- \pi \nu / \sqrt{\widetilde \mu}} 
\end{align}
with (compare with \cite[(45)]{SeiYngZag12})
\begin{equation} \label{inequality x 1}
1 \le \xi := \mathrm{e}^{\pi \nu / \sqrt{\widetilde \mu}} \int\limits_{\pi \nu / \sqrt{\widetilde \mu}}^{\infty} \left( t - \left( \dfrac{\pi \nu}{\sqrt{\widetilde \mu}} \right)^2 t^{-1} \right) \mathrm{e}^{-t} \, \mathrm{d} t \le 2\ .
\end{equation}
Thus we have
\begin{align} \label{equation for E and zeta}
\E[ N_{g_N,\widetilde \mu} ] = \zeta_N \dfrac{\widetilde \mu}{g_N \nu} 
\xi \mathrm{e}^{- \pi \nu / \sqrt{\widetilde \mu}}
\end{align}
with $2/3 \le \zeta_N \le 1$ for any $N \in \mathds N$. Since $(\widetilde \mu 
/ g_N \nu_N) \mathrm{e}^{- \pi \nu_N / \sqrt{\widetilde \mu}} $ is a continuous function 
of $\widetilde \mu$ that converges to zero for $\widetilde \mu \searrow 0$ and 
to $\infty$ for $\widetilde \mu \to \infty$, for any $N \in \mathds N$ we are 
able to choose a $\muN$ in such a way that
\begin{align}\label{determinging forumla for mu}
\begin{split}
\mathds E\left[ N_{\gN,\muN}(s_N^{-1} \cdot) \right] & = \int\limits_{\pi s_N / \sqrt{\mu_N}}^{\infty} N_{\gN,\muN}(s_N^{-1} l) \, \nu \mathrm{e}^{-\nu l} \, \mathrm{d} l = \int\limits_{\pi / \sqrt{\muN}}^{\infty} N_{\gN,\muN}(l) \, \nu_N \mathrm{e}^{- \nu_N l} \, \mathrm{d} l \\
& =: \mathds E^{\nu_N} \left[ N_{\gN,\muN} \right] =\zeta_N \dfrac{\muN}{g_N 
\nu_N} 
\xi_N \mathrm{e}^{- \pi \nu_N / \sqrt{\muN}} \stackrel{!}{=} 
\zeta_N \xi_N \rho \nu_N^{-1}
\end{split}
\end{align}
with 
\begin{equation} \label{inequality x 2}
1 \le \xiN := \mathrm{e}^{\pi \nu_N / \sqrt{\muN}} \int\limits_{\pi \nu_N / \sqrt{\muN}}^{\infty} \left( t - \left( \dfrac{\pi \nu_N}{\sqrt{\muN}} \right)^2 t^{-1} \right) \mathrm{e}^{-t} \, \mathrm{d} t \le 2\ .
\end{equation}
Hence
\begin{align} \label{inequality for E NgNmuN}
 \dfrac{2}{3} \rho \nu_N^{-1} \le \E \left[ N_{g_N, \mu_N}(s_N^{-1} \cdot) \right] 
\le 2 \rho \nu_N^{-1} \ .
\end{align}
%
%

Similar to \cite[Subsection.~3.3.3]{SeiYngZag12}, we define the occupation numbers
\begin{equation} \label{Besetzungszahlen}
M_j := \left\lceil N_{\gN, \muN}(\tilde{l}_j) \right\rceil
\end{equation}
for all intervals $\widetilde W_j$ with $\tilde{l}_j > 0$ except for $\widetilde W_0$ and the first and last interval within $\Lambda_N$ and set $M_j := 0$ else. We remark that we set $M_0 = 0$ because, unlike $|(x_j(\omega), x_{j+1}(\omega))|$ for any $j \in \mathds Z \backslash \{ 0 \}$, the length $|(x_0(\omega), x_{1}(\omega))|$ is not exponentially distributed (cf. waiting time paradox~\cite{kingman1993poisson}). We also note that $\{M_j\}_{j \in \mathds Z}$ is a general admissible sequence.

We define 
\begin{equation}\label{definition of lambda}
\lambdaN := \mathds{P}\big( l > s_N \pi / \sqrt{\mu_N} \big) = \int\limits_{s_N \pi / \sqrt{\mu_N}}^{\infty} \, \nu \mathrm{e}^{-\nu l} \, \mathrm{d} l = \mathrm{e}^{- \pi \nu_N / \sqrt{\muN}}
\end{equation}
and note that $\lambdaN$ is asymptotically equal to the fraction of intervals that are large enough to be, according to our choice~\eqref{Besetzungszahlen}, occupied by at least one particle. Moreover, after taking into account~\eqref{determinging forumla for mu} and~\eqref{equation for E and zeta} we are able to establish the relationship
\begin{equation} \label{Zusammenhang g lambda}
\gN = \rho^{-1} \muN \lambdaN = \rho^{-1} \pi^2  
\dfrac{\lambdaN \nu_N^2}{[\ln \lambdaN]^2}\ .
\end{equation}
%

%
%
%
%
%
For the proof of the next theorem we need the following fact about the length $\llargesttilde$ of the largest scaled subinterval.
\begin{lemma}\label{remark largest interval}
	For {$1/N \ll \nu_N \lesssim 1$},  for any $\kappa >4$ and for almost any $\omega \in \Omega$ there exists an $\widetilde N \in \mathds N$ such that for any $N \ge \widetilde N$ the inequality
	\begin{align} \label{equation largest interval}
	\llargesttilde \le \kappa \, \nu_N^{-1} \ln (N)
	\end{align}
	holds.
\end{lemma} 
This has been proved in~[Theorem 6.2,\cite{SeiYngZag12}] for $\nu_N \equiv const.$ and can be extended to variable $\nu_N$ by scaling. More precisely, for any $\kappa > 4$ and almost any $\omega \in \Omega$ there exists an $\widetilde N \in \mathds N$ such that for any $N \ge \widetilde N$ it is $\llargest = \max\{ |(x_j(\omega), x_{j+1}(\omega)) \cap \Lambda_N :  j \in \mathds Z\} \le \kappa \nu^{-1} \ln(N)$ and therefore
\begin{align}
\llargesttilde & = \max\{ |\widetilde W_j | : j \in \mathds Z \} = \max\{   |(s_N^{-1} x_j(\omega), s_N^{-1} x_{j+1}(\omega)) \cap \Lambda_N| : j \in \mathds Z \} \\
& = s_N^{-1}  \max\{ |(x_j(\omega), x_{j+1}(\omega)) \cap s_N \Lambda_N| : j \in \mathds Z \} = s_N^{-1} \ell^{1,>}_{s_N N}  \le \kappa \nu_N^{-1} \ln(N) \ . 
\end{align}
Finally, we establish the main theorem of this section which plays a central role in proving BEC in the following section. For its proof we will use Theorem~\ref{thm 
	ConvergenceParticleNumber almost surely} which itself is proved in the appendix. 
\begin{theorem} \label{theorem upper bound ELS}
	Let $[\ln (N)]^4 / N^{1-2\eta} \ll \nu_N \lesssim 1$  and $g_N \ll 
\nu_N^2 N^{-\eta} [\ln (N)]^{-2}$ with an $0 < \eta \le 1/3$ be given. Then, with 
$K(\eta) := 10 \cdot 5^2 \cdot \pi^2 \eta^{-2} + 1$ and $c_1(\eta) := 90 
E^{\mathrm{GP}}(1, 1, K(\eta)) \eta^{-2}$ one obtains
	\begin{align} \label{asymptotic 1 proof M_j N_j as}
	\dfrac{1}{N} E_{\mathrm{LS}}^{\mathrm{QM}} (N, \Lambda_N, g_N) \le c_1(\eta) \dfrac{\nu_N^2}{( \ln (N))^{2}}
	\end{align}
	which holds for all but finitely many $N \in \mathds N$ almost surely.
\end{theorem}
\begin{proof}
	We introduce the new interaction strength $\hat g_N:= \rho^{-1} 
[\pi^2 \nu_N^2 \eta^{-2} \ln (N)^{-2}] N^{-\eta}$ and using relation~\eqref{Zusammenhang g lambda} one obtains $\hat{\mu}_N = \pi^2 
\nu_N^2 \eta^{-2} (\ln (N))^{-2}$. Note that the hat characterizes the corresponding quantities associated with the pair $(\hat g_N,\nu_N)$. By assumption $g_N \le \hat g_N$ for all but finitely many $N \in \mathds N$ and $\hat{\lambda}_N = N^{-\eta}$. Furthermore, we choose
	$\hat M_j$ as in \eqref{Besetzungszahlen}.

	Since $\hat \lambda_N = N^{-\eta} \gg \ln (N) \ln (\nu_N N) (\nu_N N)^{-1/2}$ 
it is $1/9 \le N^{-1} \sum_{j \in \mathds Z} \hat M_j \le 9$ for all but finitely 
many $N \in \mathds N$ almost surely by Theorem~\ref{thm 
ConvergenceParticleNumber almost surely}.
	
	Defining $\widehat n_N := N / \sum_{j \in \mathds Z} \hat  M_j$ one clearly has $\sum_{j \in \mathds Z} \left\lceil \widehat n_N \hat M_j \right\rceil \ge N$ for all but finitely many $N \in \mathds N$ almost surely. Note that $\left\lceil \widehat n_N \hat M_j \right\rceil$ is a general admissible sequence. By Theorem~\ref{thm energy bound}, because $\hat M_j > 0$ if and only if $\tilde l_j > \pi / \sqrt{\hat \mu_N}$ and with \eqref{E0QM Mj ge N},
	\begin{align} 
	\dfrac{1}{N} E_{\mathrm{LS}}^{\mathrm{QM}}(N, \Lambda_N, g_N) & \le \dfrac{1}{N} \sum\limits_{j \in \mathds Z} E_0^{\mathrm{QM}} \left( \left\lceil \widehat n_N \hat M_j \right\rceil, \tilde l_j, g_N \right) \\
	& \le \dfrac{1}{N} \sum\limits_{j \in \mathds Z \, : \, \tilde l_j > \pi / \sqrt{\hat \mu_N}} \dfrac{\left\lceil \widehat n_N \hat M_j \right\rceil}{\tilde l_j^2} E^{\mathrm{GP}} \left( 1, 1, \left\lceil \widehat n_N \hat M_j \right\rceil \tilde l_j g_N \right) \\
	& \le \dfrac{\mu_N}{\pi^2} \dfrac{1}{N} \sum\limits_{j \in \mathds Z \, : \, \tilde l_j > \pi / \sqrt{\hat \mu_N}} \left\lceil \widehat n_N \hat M_j \right\rceil E^{\mathrm{GP}} \left( 1, 1, \left\lceil \widehat n_N \hat M_j \right\rceil \tilde l_j g_N \right)
	\end{align}
	for all but finitely many $N \in \mathds N$ almost surely. 
	
	Finally, we have $\lceil \widehat n_N \hat M_j \rceil \le ( 9\hat M_j 
+ 1 ) \le 10 \hat M_j$ and hence $\lceil \widehat n_N \hat M_j \rceil \tilde 
l_j g_N \le 10 \hat M_j \tilde l_j \hat g_N \le 10 \hat \mu_N (\llargesttilde)^2 + 
10 \tilde l_j \hat g_N \le 10 \cdot 5^2 \pi^2 \eta^{-2} + 1 = K(\eta)$ 
for any $j \in \mathds Z$. This yields
	\begin{align}
	& \dfrac{\nu_N^2}{\eta^2 [\ln (N)]^2}\dfrac{1}{N} \sum\limits_{j \in \mathds Z \, : \, \tilde l_j > \pi / \sqrt{\hat \mu_N}} \left\lceil \widehat n_N M_j \right\rceil E^{\mathrm{GP}} \left( 1, 1, \left\lceil \widehat n_N M_j \right\rceil \tilde l_j g_N \right) \\
	\le \, & \dfrac{c_1(\eta)}{9} \dfrac{\nu_N^2}{[\ln (N)]^2} \dfrac{1}{N} 
\sum\limits_{j \in \mathds Z} \hat M_j \le c_1(\eta) \dfrac{\nu_N^2}{[\ln (N)]^2}
	\end{align}
	for all but finitely many $N \in \mathds N$ almost surely.
\end{proof}
%

%
\section{Main results}
Generalized BEC in non-interacting Bose gases is said to occur (almost surely) if, for the sequence $(\phi_N^j)_{j \in \mathds Z}$ of eigenstates with respective occupation numbers $n_{\phi_N^j}$ (see Definition~\ref{definition occupation number}) and energies $E_N^j \geq 0$ one has that
\begin{equation}
\lim\limits_{\epsilon \searrow 0} \limsup\limits_{N \to \infty} \dfrac{1}{N} \sum\limits_{j \in \mathds Z : E_N^j \le \epsilon} n_{\phi_N^j}
\end{equation}
is (almost surely) larger than zero~\cite{van1982generalized}. By analogy, we say that \textit{generalized BEC} occurs almost surely in our model, the Luttinger--Sy model with interaction, if
\begin{equation} \label{definition generalized BEC}
\rho_0  := \lim\limits_{\epsilon \searrow 0} \limsup\limits_{N \to \infty} \dfrac{1}{N} \sum\limits_{j \in \mathds Z : E^{\mathrm{GP}}(1, \tilde l_j, N_j g_N) \le \epsilon} n_{N_j} > 0
\end{equation}
is almost surely larger than zero. Here $n_{N_j}$ is the number of particles occupying the single-particle state $\phi_{\tilde l_j, N_j g_N}^{\mathrm{GP}}$, see Definition~\eqref{definition occupation number}. Consequently, we refer to the case
\begin{equation}
\lim\limits_{\epsilon \searrow 0} \lim\limits_{N \to \infty} \dfrac{1}{N} \sum\limits_{j \in \mathds Z : E^{\mathrm{GP}}(1, \tilde l_j, N_j g_N) \le \epsilon} n_{N_j} = 1
\end{equation}
almost surely as almost sure \textit{complete} generalized BEC.

In general, a single-particle state $\phi$ with respective occupation number $n_{\phi}$ is called almost surely \textit{macroscopically occupied} if almost surely $\limsup_{N \to \infty} n_{\phi} / N > 0$. In this paper we refer to almost sure type-I (type-II) BEC if finitely (infinitely) many states $\phi_{\tilde l_j, N_j g_N}^{\mathrm{GP}}$ are almost surely macroscopically occupied. If for almost any $\omega \in \Omega$ one has $\rho_0 > 0$ without any state $\phi_{\tilde l_j, N_j g_N}^{\mathrm{GP}}$ being macroscopically occupied we speak of almost sure type-III BEC.
\begin{theorem}\label{Theorem generalized BEC}[Generalized BEC]
	Assume that $[\ln (N)]^4 / N^{1-2\eta} \ll \nu_N \lesssim 1$  and $g_N \ll \nu_N^2 N^{-\eta} [\ln (N)]^{-2}$ with $0 < \eta \le 1/3$. Then almost sure 
	complete generalized BEC occurs.
\end{theorem}
\begin{proof}
	For any $N \in \mathds N$ let $\{N_j\}_{j \in \mathds{Z}}$ be a sequence of occupation numbers of the intervals with respect to the ground state of the Luttinger--Sy model, i.e. one has $E_{\mathrm{LS}}^{\mathrm{QM}}(N, \Lambda_N, g_N) = \sum_{j \in \mathds Z} E_0^{\mathrm{QM}}(N_j, \tilde l_j, g_N)$. We will show that
	\begin{align}
	\rho_0 = \lim\limits_{\epsilon \searrow 0} \liminf\limits_{N \to \infty} \dfrac{1}{N} \sum_{j \in \mathds Z : E^{\mathrm{GP}}(1, \tilde l_j, N_j g_N) \le \epsilon} n_{N_j}  =  1
	\end{align}
which then proves the statement.
	
	 Recall that according to Lemma~\ref{remark largest interval}, there exists a set $\widehat \Omega \subset \Omega$ with $\mathds{P}(\widehat \Omega) = 1$ and the following property: for any $\omega \in \widehat \Omega$ there exists an $\widetilde N_1(\omega) \in \mathds N$ such that for any $N \ge \widetilde N_1(\omega)$ one has $\llargesttilde \le 5 \nu_N^{-1} \ln (N)$. Moreover, for any $\omega \in \Omega$ we define $\widehat N_1(\omega)$ such that $\widehat N_1(\omega) \ge \widetilde N_1(\omega)$ and $g_N \le \nu_N^2 N^{-\eta} [\ln (N)]^{-2}$ for any $N \ge \widehat N_1(\omega)$.
	
	In a first step we show that for any $\omega \in \widehat \Omega$, any $j \in \mathds Z$ as well as any  $N > \widehat N_2(\omega) := \max\{3,\widehat N_1(\omega), (40(\tilde{c}+\sqrt{2}c)^{8/3} \| \nu_N \|_{\infty} )^{1/ \eta}\}$ with $\|\nu_N\|_{\infty} = \max\{ \nu_N : N \in \mathds N\}$ the inequality \linebreak $\gamma_j := N_j \tilde l_j g_N < [\ln (N)]^{1/2}$ implies $N_j^{1/3} \tilde l_j g_N < 1/(4(\tilde{c}+\sqrt{2}c)^2)$, $\tilde{c}$ and $c$ from Lemma~\ref{lemma 1} : To do this let $\omega \in \widehat \Omega$ and $N > \widehat N_2(\omega)$ be given. If $N_j \ge 4^{3/2}\tilde{c}^2 (\ln (N))^{3/4}$ then
	\begin{align}\label{EqProofI}
	N_j^{1/3} \tilde l_j g_N = N_j^{- 2/3}N_j \tilde l_j g_N < N_j^{- 2/3} [\ln (N)]^{1/2} \le \dfrac{1}{4(\tilde{c}+\sqrt{2}c)^2} \ .
	\end{align}
	On the other hand, if $N_j < 4^{3/2}\tilde{c}^2 (\ln (N))^{3/4}$ then
	\begin{align}\label{EqProofII}
	N_j^{1/3} \tilde l_j g_N < 2\tilde{c}^{2/3} (\ln (N))^{1/4} \llargesttilde \nu_N^2 N^{-\eta} [\ln (N)]^{-2} \le \dfrac{1}{4(\tilde{c}+\sqrt{2}c)^2}\ .
	\end{align}
	%
	
	We now prove that
	\begin{align} \label{Njtrue convergence to 1}
	\lim\limits_{N \to \infty} \dfrac{1}{N} \sum\limits_{j \in \mathds Z : \gamma_j < (\ln (N))^{1/2} } N_j = 1
	\end{align}
	almost surely: Suppose there exists a set $\widetilde \Omega \subset \Omega$ with $\mathds{P}(\widetilde \Omega) > 0$ such that for any $\omega \in \widetilde \Omega$ there is a $\widetilde d > 0$ with $\limsup_{N \to \infty} N^{-1} \sum_{j \in \mathds Z :  \gamma_j \ge (\ln (N))^{1/2} } N_j \ge \widetilde d$. Then for any $\omega \in \widetilde \Omega \cap \widehat \Omega$, using Theorem~\ref{thm energy bound} (with $\epsilon = 1/2$), equation \eqref{EGP scaling}, inequality \eqref{bound EGP} while setting $c_2 : = \widetilde d / (8 \cdot 5^{2})$, we obtain
	\begin{align}
	\dfrac{1}{N} E_{\mathrm{LS}}^{\mathrm{QM}} (N, \Lambda_N, g_N) & \ge \dfrac{1}{N} \sum\limits_{j \in \mathds Z : \gamma_j \ge (\ln (N))^{1/2}} E_0^{\mathrm{QM}}(N_j, \tilde l_j, g_{N}) \\
	& \ge \dfrac{1}{{N}} \sum\limits_{j \in \mathds Z : \gamma_j \ge (\ln {N})^{1/2} } \dfrac{1}{2} N_j \dfrac{N_j g_N}{2 \tilde l_j} \ge c_2 \dfrac{\nu_N^2}{(\ln (N))^{3/2}} 
	\end{align}
	for infinitely many $N \in \mathds N$. However, since $\mathds{P}(\widetilde \Omega \cap \widehat \Omega) = \mathds{P}(\widetilde \Omega) > 0$ this is in contradiction with Theorem~\ref{theorem upper bound ELS}.
	
	Next, we prove that for any $\omega \in \widehat \Omega$ and any $N \ge \max\{\widehat N_2(\omega), \mathrm{e}^{1/\pi^2}, (5 \| \nu_N \|_{\infty})^{2/\eta},\widehat{N}_3(\omega)\}$  one has
	\begin{align}\label{third step gen BEC}
	& \dfrac{1}{N} \sum\limits_{j \in \mathds Z : \gamma_j < (\ln (N))^{1/2} } n_{N_j} \ge \left( 1 - \dfrac{c_3}{\ln (N)} \right) \dfrac{1}{N} \sum\limits_{j \in \mathds Z : \gamma_j < (\ln (N))^{1/2}} N_j
	\end{align}
	with $c_3 := \sqrt{7} c \left( 1 + \mathrm{e}^{-2 \pi} \right) \mathrm{e}^{4 \pi}$ and $c > 0$ as in Theorem~\ref{Theorem 2}: By Theorem~\ref{Theorem 2} and  using that $\ln(1 + x) \ge x / (1 + x)$ for $x > - 1$ one infers that if $\gamma_j < (\ln (N))^{1/2}$ and $N_j \ge 1$ for some $j\in \mathbb{Z}$ then
	\begin{align}
	\left( 1 - \dfrac{n_{N_j}}{N_j} \right) & \le \sqrt{7} c \cdot \left\{ \ln \left[  1 + \mathrm{e}^{-2 (\pi^2 + 3 N_j \tilde l_j g_N)^{1/2}} \right] \right\}^{-1} N_j^{2/3} \tilde l_j g_N \\
	& \le \sqrt{7} c \dfrac{1 + \mathrm{e}^{-2 \pi}}{\mathrm{e}^{-2 \sqrt{\pi^2 + 3 \gamma_j}}} N_j^{-1/3} \gamma_j \ .
	\end{align}
	On the one hand, for any $j \in \mathds Z$ with $\gamma_j \le (\ln (N))^{-1}$ it is
	\begin{align}
	\sqrt{7} c \left( 1 + \mathrm{e}^{-2 \pi} \right) \mathrm{e}^{2 \sqrt{\pi^2 + 3 \gamma_j}} N_j^{-1/3} \gamma_j & \le \sqrt{7} c \left( 1 + \mathrm{e}^{-2 \pi} \right) \mathrm{e}^{4 \pi} \dfrac{1}{\ln (N)} = \dfrac{c_3}{\ln (N)} \ .
	\end{align}
	On the other hand, for any $j \in \mathds Z$ with $(\ln (N))^{-1} \le \gamma_j < (\ln (N))^{1/2}$ we have
	\begin{align}
	N_j N^{-\eta/2} \ge N_j \llargesttilde g_N \ge N_j \tilde l_j g_N = \gamma_j \ge \dfrac{1}{\ln (N)} \ ,
	\end{align}
	since $\llargesttilde \le 5 \nu_N^{-1} \ln (N)$ and $N \ge (5 \| \nu_N \|_{\infty})^{2/\eta}$ imply $g_N \le (N^{\eta/2} \llargesttilde)^{-1}$. Thus $N_j \ge N^{\eta/2} / \ln (N)$ which leads to 
	\begin{align}
	& \sqrt{7} c \left( 1 + \mathrm{e}^{-2 \pi} \right) \mathrm{e}^{2 \sqrt{\pi^2 + 3 \gamma_j}} N_j^{-1/3} \gamma_j \le \dfrac{c_3}{\ln (N)}
	\end{align}
	for all $N \geq \widehat{N}_3(\omega)$, $\widehat{N}_3(\omega)$ some constant.
	Hence, 
	we have $( 1 - n_{N_j}/N_j ) \le c_3/\ln (N)$ or, equivalently, $(1 - [c_3/\ln (N)]) N_j \le n_{N_j}$ which implies~\eqref{third step gen BEC}.
	
	Our last step is to show that for any $\epsilon > 0$
	\begin{align}\label{nNj E_j ge epsilon to zero}
	\lim\limits_{N \to \infty} \dfrac{1}{N} \sum\limits_{j \in \mathds Z : \gamma_j < (\ln (N))^{1/2}, E_j \ge \epsilon} n_{N_j} =  0
	\end{align}
	almost surely with $E_j := E^{\mathrm{GP}}(1, \tilde l_j, N_j g_N)$: We assume to the contrary that there exist an $\epsilon > 0$ and a set $\widetilde \Omega \subset \Omega$ with $\mathds{P}(\widetilde \Omega) > 0$ such that for any $\omega \in \widetilde \Omega$ there is a constant $\widetilde r > 0$ with $\limsup_{N \to \infty} N^{-1} \sum_{j \in \mathds Z : \gamma_j < (\ln (N))^{1/2}, E_j > \epsilon} n_{N_j} \ge \tilde r$. Then for any $\omega \in \widetilde \Omega$ one also has
	\begin{align}
	\limsup\limits_{N \to \infty} \dfrac{1}{N} \sum\limits_{j \in \mathds Z : \gamma_j < (\ln (N))^{1/2}, E_j > \epsilon} N_j \ge \limsup\limits_{N \to \infty} \dfrac{1}{N} \sum\limits_{j \in \mathds Z : \gamma_j < (\ln (N))^{1/2}, E_j > \epsilon} n_{N_j} \ge \tilde r 
	\end{align}
	since $N_j \ge n_{N_j}$ for any $j \in \mathds Z$, $N \in \mathds N$. Hence, with Theorem~\ref{thm energy bound} ($\epsilon = 1/2$), equation \eqref{EGP scaling} and inequality \eqref{bound EGP},
	\begin{align}\label{EqProofIII}
	\dfrac{1}{N} E_{\mathrm{LS}}^{\mathrm{QM}} (N, \Lambda_N, g_N) & \ge \dfrac{1}{N} \sum\limits_{j \in \mathds Z : \gamma_j < (\ln (N))^{1/2}, E_j > \epsilon} E_0^{\mathrm{QM}}(N_j, \tilde l_j, g_N) \\
	& \ge \dfrac{1}{N} \sum\limits_{j \in \mathds Z : \gamma_j < (\ln (N))^{1/2}, E_j > \epsilon} \dfrac{1}{2} E^{\mathrm{GP}}(N_j, \tilde l_j, g_N) \\
	& \ge \dfrac{\epsilon}{2} \dfrac{1}{N} \sum\limits_{j \in \mathds Z : \gamma_j < (\ln (N))^{1/2}, E_j > \epsilon} N_j \ge \dfrac{\epsilon}{2} \tilde r
	\end{align}
	for infinitely many $N \in \mathds N$ which is again a contradiction to Theorem~\ref{theorem upper bound ELS}. Note that the assumptions of Theorem~\ref{thm energy bound} ($\epsilon = 1/2$) are fulfilled according to~\eqref{EqProofI} and~\eqref{EqProofII}.
	
	Altogether we have shown that, using \eqref{nNj E_j ge epsilon to zero}, \eqref{third step gen BEC} and \eqref{Njtrue convergence to 1} respectively,
	\begin{align}
	\lim_{\epsilon \searrow 0} \liminf\limits_{N \to \infty} \dfrac{1}{N} \sum\limits_{j \in \mathds Z : E_j \le \epsilon} n_{N_j} & \ge \lim_{\epsilon \searrow 0} \liminf\limits_{N \to \infty} \dfrac{1}{N} \sum\limits_{j \in \mathds Z : \gamma_j < (\ln (N))^{1/2} , E_j \le \epsilon} n_{N_j} \\
	& = \lim_{\epsilon \searrow 0} \liminf\limits_{N \to \infty} \dfrac{1}{N} \sum\limits_{j \in \mathds Z : \gamma_j < (\ln (N))^{1/2}} n_{N_j} \\
	& \ge \liminf\limits_{N \to \infty} \left( 1 - \dfrac{c_3}{\ln (N)} \right) \dfrac{1}{N}  \sum\limits_{j \in \mathds Z : \gamma_j < (\ln (N))^{1/2} } N_j =  1
	\end{align}
	almost surely.
\end{proof}
\begin{remark} Whereas in \cite{SeiYngZag12} type-I BEC in probability is shown in the regime where $1/N \ll \nuN \ll [\ln (N)]^3 / N$, see Appendix \ref{Appendix Seiringer gN nuN limit}, we are able to allow for $\nu_N \equiv (const.)$ which is mainly due to two reasons: Firstly, we consider BEC in the generalized sense. Having proved now complete BEC in a generalized sense we may
	replace the $\limsup$ in Definition \eqref{definition generalized BEC} by $\lim$. Secondly, instead of the whole system length $\LN$ and particle number $N$, we established Theorem~\ref{thm energy bound} and Theorem~\ref{Theorem 2} containing the occupation numbers $N_j$ and lengths $\tilde{l}_j$ of the individual intervals. This eventually enables us to use stronger interactions than in \cite{SeiYngZag12} in the sense that $\gN$ converges to zero more slowly, see again Appendix \ref{Appendix Seiringer gN nuN limit}.
\end{remark}
\begin{theorem}[Transition of condensation] \label{Theorem transition of condensation} Let $0 < \eta \le 1/3$ and $[\ln (N)]^4 / N^{1-2\eta} \ll \nu_N \lesssim 1$ be given. Then, 1) if  $g_N \equiv 0$ then almost surely exactly one single-particle state is macroscopically occupied and hence complete type-I BEC occurs, 2) if $g_N \ll \nu_N N^{-1} [ \ln (N)]^{-2}$ then BEC is almost surely of type I or II, 3) BEC is of type-III almost surely if $\nuN N^{-1}  [ \ln (N)]^{-1} \ll g_N \ll \nuN^{2} N^{-\eta} [\ln (N)]^{-2}$.
\end{theorem}
\begin{proof}
	Assume that $\{N_j\}_{j \in \mathds Z}$ are occupation numbers of the intervals with respect to the ground state of the Luttinger--Sy model, i.e. $E_{\mathrm{LS}}^{\mathrm{QM}}(N, \Lambda_N, g_N) = \sum_{j \in \mathds Z} E_0^{\mathrm{QM}}(N_j, \tilde l_j, g_N)$ for any $N \in \mathds N$. As before, $n_{N_j}$ denotes the number of particles occupying the single-particle state $\phi_{\tilde l_j, N_j g_N}^{\mathrm{GP}}$. 
	
	The first part of the theorem regarding the case $g_N \equiv 0$ follows readily since almost surely there is only one largest interval and, since the temperature is zero, all particles occupy the ground state corresponding to this length.
	
	Next, we treat the case $g_N \ll \nu_N N^{-1} [ \ln (N)]^{-2}$: According to Corollary~\ref{Corollary E-Gap nuN} there exists, for any
	$\eta' > 0$ and any $C_3 > 2\mathrm{e}^{\nu/\rho}$, an $\widetilde N(\eta') \in \mathds N$ such that for $N \ge \widetilde N(\eta')$ it is
	\begin{equation}\begin{split}
	\mathds{P} (\Omega_ 1)  > 1 - \eta'
	\end{split}
	\end{equation}
	with 
	\begin{align*}
	\Omega_ 1:=\left\{\omega \in \Omega:\llargesttilde \ge (1/2) \nu_N^{-1} \ln(\nu_N N) \, , \,  \llargesttilde - \tilde l^{>,\left\lceil 2\nu C_3 / (C_1 \eta') \right\rceil + 1}_N > \nu_N^{-1} \ln(C_3/(2\mathrm{e}^{\nu/\rho}))    \right\}
	\end{align*}
	and $C_1 := - \nu / [4 \ln(\eta' / 2)]$. Also, $\tilde l^{>,\left\lceil 2\nu C_3 / (C_1 \eta') \right\rceil + 1}_N$ is the $(\left\lceil 2\nu C_3 /( C_1 \eta') \right\rceil + 1)$th largest length of $\{\tilde l_j\}_{j \in \mathds Z}$. Note that if $\omega \in \Omega_1$ then there are at most $(\left\lceil 2\nu C_3 / (C_1 \eta') \right\rceil + 1)$ many intervals that have a length larger than $\hat l_N := \llargesttilde - \nu_N^{-1}  \ln(C_3/(2\mathrm{e}^{\nu/\rho}))$.
	
	For convenience we define $E^{\mathrm{QM}}_0(0, 0 , g) := 0$ for $g \ge 0$. Furthermore, for any $N \in \mathds N$ and $j \in \mathds Z$ we write $\widehat N^{(1)} := \sum_{j \in \mathds Z : \tilde l_j \ge \hat l_N } N_j$ and $\widehat N^{(2)} := \sum_{j \in \mathds Z : \tilde l_j < \hat l_N } N_j$. Let $\widetilde{N}^{>}_N$ denote the particle number in the largest interval.
	
	Now, let $\omega \in \Omega_1$ be given. Then
	\begin{align}
	\begin{split} \label{equation 2 proof type I}
	& E_{\mathrm{LS}}^{\mathrm{QM}}(N,\Lambda_N, g_N) = \sum\limits_{j \in \mathds Z } E^{\mathrm{QM}}_0(N_j, \tilde l_j , g_N) \\
	\ge \, & \sum\limits_{j \in \mathds Z : \tilde l_j \ge \hat l_N, \tilde l_j \neq \llargesttilde } E^{\mathrm{QM}}_0(N_j, \tilde l_j , g_N) + E^{\mathrm{QM}}_0(\widetilde{N}^{>}_N, \llargesttilde , 0) + E^{\mathrm{QM}}_0(\widehat N^{(2)} , \hat l_N , 0) \\
	= \, & \sum\limits_{j \in \mathds Z : \tilde l_j \ge \hat l_N, \tilde l_j \neq \llargesttilde } E^{\mathrm{QM}}_0(N_j, \tilde l_j , g_N) + \widetilde{N}^{>}_N \dfrac{\pi^2}{( \llargesttilde )^2} +  \widehat N^{(2)} \dfrac{\pi^2}{\hat l_N^2} \ .
	\end{split}
	\end{align}
	On the other hand, with $\phi_0$ the ground state of $H(1, 1,0)$, $C := \int_0^1 | \phi_0|^4$, we can employ a simple variational argument to obtain (see also proof of Theorem~\ref{thm energy bound})
	\begin{align}\label{equation proof type I}
	\begin{split}
	E^{\mathrm{QM}}_0(\widetilde{N}^{>}_N + \widehat N^{(2)}, \llargesttilde , g_N) & \le E^{\mathrm{GP}}(\widetilde{N}^{>}_N + \widehat N^{(2)}, \llargesttilde , g_N) \\ 
	& \le \dfrac{(\widetilde{N}^{>}_N + \widehat N^{(2)})}{(\llargesttilde)^2} \mathcal E^{\mathrm{GP}}(1 , 1 , (\widetilde{N}^{>}_N + \widehat N^{(2)} ) \llargesttilde g_N)[\phi_0] \\
	& = \dfrac{(\widetilde{N}^{>}_N + \widehat N^{(2)})}{(\llargesttilde)^2} \pi^2 + \dfrac{(\widetilde{N}^{>}_N + \widehat N^{(2)})^2 g_N}{2\llargesttilde} C \ .
	\end{split}
	\end{align}
	
	According to \eqref{equation largest interval}, $[\llargesttilde - \kappa \nu_N^{-1} \ln (N)]_{+}$ converges almost surely to zero and hence in probability to zero for any $\kappa > 4$. Therefore, for any $\eta' > 0$ there exists an $\widehat N(\eta') \in \mathds N$ such that for any $N \ge \widehat N(\eta')$ it is $\mathds{P}(\Omega_2) := \mathds{P} (\llargesttilde < 5 \nu_N^{-1} \ln (N) ) > 1 - \eta'$.
	
	Let $\epsilon, \eta', \eta > 0$ and $\omega \in \Omega_1 \cap \Omega_2$ be given. Note that $\mathds{P}(\Omega_1 \cap \Omega_2) > 1 - 2 \eta'$ for any $N \ge \max\{ \widetilde N(\eta'), \widehat N(\eta') \}$. Moreover, let $N \ge \max\{ \widetilde N(\eta'), \widehat N(\eta') \}$ be such that
	\begin{align}
	g_N \le \dfrac{2 \pi^2 \eta}{C} \dfrac{\nu_N^{-1} \ln(C_3/(2\mathrm{e}^{\nu/\rho}))}{(5 \nu_N^{-1} \ln (N))^2} \dfrac{1}{N} 
	\end{align}
	and 
	\begin{align}
	E^{\mathrm{GP}}(1,\hat l_N, N g_N) & \le \dfrac{1}{\hat l_N^2} \mathcal E^{\mathrm{GP}}(1,1, N \hat l_N g_N)[\phi_0] = \dfrac{1}{\hat l_N^2} \left( \pi^2 + \dfrac{N \hat l_N g_N}{2} C \right) \le \epsilon
	\end{align}
hold. 

\newpage

Now we assume that $\widehat N^{(2)}/ N > \eta$ which will lead to a contradiction with the above. Since $\widetilde{N}^{>}_N + \widehat N^{(2)} \le N$ we conclude that, for $N \in \mathds N$ large enough, 
	\begin{align}
	\dfrac{N}{\widehat N^{(2)}} \dfrac{(\widetilde{N}^{>}_N + \widehat N^{(2)}) g_N}{2 \pi^2 \llargesttilde} C < \dfrac{1}{2\pi^2 \eta} \dfrac{N g_N}{\llargesttilde} C \le \dfrac{\nu_N^{-1} \ln(C_3/(2\mathrm{e}^{\nu/\rho}))}{(\llargesttilde)^3}  &\le \dfrac{(\llargesttilde)^2 -  (\hat l_N)^2}{\hat l_N^2 (\llargesttilde)^2} \\
	& = \dfrac{1}{\hat l_N^2} - \dfrac{1}{(\llargesttilde)^2}
	\end{align}
	and therefore
	\begin{align}
	\widehat N^{(2)} \dfrac{\pi^2}{(\llargesttilde)^2} + \dfrac{(\widetilde{N}^{>}_N + \widehat N^{(2)})^2 g_{N}}{2 \llargesttilde } C & < \widehat N^{(2)} \dfrac{\pi^2}{\hat l_N^2} \ .
	\end{align}
	Comparing this with \eqref{equation 2 proof type I} and \eqref{equation proof type I} we arrive at a contradiction since $E_{\mathrm{LS}}^{\mathrm{QM}}(N,\Lambda_N, g_N)$ is minimal.
	
	From the assumptions $\lim_{N \to \infty} N \llargesttilde g_N = 0$ almost surely and with Theorem~\eqref{Theorem 2}, $\max_{j \in \mathds Z : N_j > 0} \{ 1 - n_{N_j} /N_j \}$ converges to zero almost surely. Therefore, for any $\epsilon, \eta, \eta'>0$ and $C_3 > 2\mathrm{e}^{\nu/\rho}$ there exists an $\widetilde N \in \mathds N$ such that for any $N \ge \widetilde N$ it is, with $E_j = E^{\mathrm{GP}}(1,\tilde l_j, N_j g_N)$,
	\begin{align}
	& \mathds{P} \left( \left\{ \left| \dfrac{1}{N} \sum_{j \in \mathds Z : \tilde l_j \ge \hat l_N, E_j \le \epsilon} N_j  - 1 \right| \le \eta \right\} \cap \Omega_1 \cap \Omega_2 \right)=\mathds{P} \left(\Omega_1 \cap \Omega_2\right) > 1 - 2 \eta' 
	\end{align}
	and
	\begin{align}
	& \mathds{P} \left( \left\{ \left| \dfrac{1}{N} \sum_{j \in \mathds Z : \tilde l_j \ge \hat l_N, E_j \le \epsilon} \left( N_j-n_{N_j}\right) \right| \le \eta \right\} \cap \Omega_1 \cap \Omega_2 \right) \\
	\ge \, & \mathds{P} \left( \left\{ \max_{j \in \mathds Z : N_j > 0} \left\{1 -  \dfrac{n_{N_j}}{N_j} \right\} \dfrac{1}{N} \sum_{j \in \mathds Z : \tilde l_j \ge \hat l_N, E_j \le \epsilon} N_j \le \eta \right\} \cap \Omega_1 \cap \Omega_2 \right) > 1 - 3
	\eta' \ .
	\end{align}
	Hence, by the previous two inequalities,
	\begin{align}
	& \mathds{P} \left( \dfrac{n_{\Nlargest}}{N} \ge \frac{1}{2}\dfrac{1 - 2 \eta}{\lceil 2\nu C_3 / (C_1 \eta') \rceil + 1} \right) \ge \mathds{P} \left( \max_{j \in \mathds Z : \tilde l_j \ge \hat l_N, E_j \le \epsilon}\left\{\dfrac{n_{N_j}}{N}\right\} \ge \dfrac{1 - 2 \eta}{\lceil 2\nu C_3 / (C_1 \eta') \rceil + 1} \right) \nonumber \\
	\ge \, & \mathds{P} \left( \left\{ \left| \dfrac{1}{N} \sum_{j \in \mathds Z : \tilde l_j \ge \hat l_N, E_j \le \epsilon} n_{N_j}  - 1 \right| \le 2 \eta \right\} \cap \Omega_1 \cap \Omega_2 \right)\label{equation 3 proof type I} \\ 
	\nonumber\\
	\nonumber\\
	\ge \, & \mathds{P} \left( \left\{ \left| \dfrac{1}{N} \sum_{j \in \mathds Z : \tilde l_j \ge \hat l_N, E_j \le \epsilon} \left( n_{N_j}  - N_j \right) \right| \le \eta \right\} \cap \left\{ \left| \dfrac{1}{N} \sum_{j \in \mathds Z : \tilde l_j \ge \hat l_N, E_j \le \epsilon} N_j  - 1 \right| \le \eta \right\} \right.\nonumber \\
	& \qquad \cap \Omega_1 \cap \Omega_2 \Bigg) > 1 - 5 \eta'\nonumber \ .
	\end{align}
	Now, suppose there exists a $0 < C \leq 1$ such that $\mathds{P}(\lim_{N \to \infty} n_{\Nlargest} / N = 0) \ge C > 0$. Then
	\begin{equation}\label{EquationIProofxxx}\begin{split}
	\limsup\limits_{N \to \infty} \mathds{P}\left( \dfrac{n_{\Nlargest}}{N} \ge \dfrac{1}{4}\dfrac{1}{\lceil 12 \nu C_3 / (C_1 C) \rceil + 1} \right) & \le \mathds{P} \left( \limsup\limits_{N \to \infty} \dfrac{n_{\Nlargest}}{N} \ge \dfrac{1}{4} \dfrac{1}{\lceil 12 \nu C_3 / (C_1 C) \rceil + 1} \right) \\
	& \le \mathds{P} \left( \neg \left( \lim\limits_{N \to \infty} \dfrac{n_{\Nlargest}}{N} = 0 \right) \right) \le 1 - C \ .
	\end{split}
	\end{equation}
	However, we have shown above that, with $\eta = 1/4$ and $\eta' = C/6$, there exists an $\widetilde N \in \mathds N$ such that for any $N \ge \widetilde N$ it is
	\begin{align}\label{EquationIIProofxxx}
	\mathds{P} \left( \dfrac{n_{\Nlargest} }{N}  \ge \dfrac{1}{4} \dfrac{1}{\lceil 12\nu C_3 / (C_1 C) \rceil + 1} \right) \ge 1 - \dfrac{5C}{6} \ ,
	\end{align}
	see~\eqref{equation 3 proof type I}. Comparing~\eqref{EquationIIProofxxx} with~\eqref{EquationIProofxxx} one arrives at a contradiction. Hence \linebreak $\mathds{P}(\lim_{N \to \infty} n_{\Nlargest} / N = 0) = 0$ and consequently $\mathds{P}(\limsup_{N \to \infty} n_{\Nlargest} / N > 0) = 1$.
	
	Now we prove the last part of the theorem: We assume to the contrary that there exists a set $\widetilde \Omega \subset \Omega$ with $\mathds{P}(\widetilde \Omega) > 0$ such that for any $\omega \in \widetilde \Omega$ there is a $\widetilde c > 0$ with $\limsup_{N \to \infty} \max_{j \in \mathds Z }\left\{n_{N_j}/N\right\} \ge \tilde c$. 
	
	Since $N_j \ge n_{N_j}$ for any $N \in \mathds N$ and $j\in \mathds Z$ one also has $\limsup_{N \to \infty} \left\{\Nlargest/N\right\} \ge \tilde c$. Thus with Theorem~\ref{thm energy bound} ($\epsilon = 1/2$), inequality \eqref{bound EGP} and $\widehat g_N := \min\{ g_N, \nu_N N^{-1/3} [\ln (N)]^{-2} \}$, for any $\omega \in \widetilde \Omega \cap \widehat \Omega$ (see beginning of proof of Theorem~\ref{Theorem generalized BEC} for definition of the set $\widehat \Omega$) it holds
	\begin{equation}\label{Proof BEC III in prob}\begin{split}
	 \dfrac{1}{N} E_{\mathrm{LS}}^{\mathrm{QM}}(N, \Lambda_N, g_N) &\ge \dfrac{1}{N} E_0^{\mathrm{QM}}(\Nlargest,\llargesttilde, g_N)\\ &\ge \dfrac{1}{N} E_0^{\mathrm{QM}}(\Nlargest,\llargesttilde, \widehat g_N) \\
	& \ge \dfrac{1}{2} \dfrac{1}{N} E^{\mathrm{GP}}(\Nlargest, \llargesttilde, \widehat g_{N})\\  &\ge \dfrac{1}{4} \dfrac{1}{N} \dfrac{\left( \Nlargest \right)^2 \widehat g_N}{\llargesttilde} \ge \dfrac{1}{20} \dfrac{1}{N} \dfrac{\left( (\tilde c/2) N \right)^2 \widehat g_N}{\nu_N^{-1} \ln (N)}  
	\end{split}
	\end{equation}
	for infinitely many $N \in \mathds N$. Note that we inserted the length of the largest interval in the first step which is possible since the energy goes down when increasing the length. 
	
	However, since $\widehat g_N \gg \nu_N (\ln (N))^{-1} N^{-1}$ and $\mathds{P}(\widetilde \Omega \cap \widehat \Omega) = \mathds{P}(\widetilde \Omega) > 0$, this contradicts Theorem~\ref{theorem upper bound ELS}. Therefore almost surely $\lim_{N \to \infty} \max_{j \in \mathds Z }\left\{n_{N_j}/N\right\}= 0$ and hence the statement follows with Theorem~\ref{Theorem generalized BEC}.
\end{proof}

\begin{theorem}[Particle density in largest interval]\label{TheoremParticleDensity}
	Let $[\ln (N)]^4 / N^{1-2\eta} \ll \nu_N \lesssim 1$  and $g_N \ll \nu_N^2 N^{-\eta} [\ln (N)]^{-2}$ with an $0 < \eta \le 1/3$ be given. Moreover, for any $N \in \mathds N$ let $\{N_j\}_{j \in \mathds Z}$ be occupation numbers of the intervals with respect to the ground state of the Luttinger--Sy model, i.e., $E_{\mathrm{LS}}^{\mathrm{QM}}(N, \Lambda_N, g_N) = \sum_{j \in \mathds Z} E_0^{\mathrm{QM}}(N_j, \tilde l_j, g_N)$. 
	
	Then, with $c_1(\eta)$ as in Theorem~\ref{theorem upper bound ELS} one has
	\begin{align}
	\dfrac{\Nlargest}{\llargesttilde} \le \sqrt{ 8 c_1(\eta)} \dfrac{\nu_N^{3/2}}{\ln(N)} \dfrac{1}{[\ln(\nu_N N)]^{1/2}} \left( \min\left\{ g_N, \dfrac{\nu_N}{N^{1/3} [\ln (N)]^{2} } \right\} \right)^{- 1/2} N^{1/2}
	\end{align}
	for all but finitely many $N \in \mathds N$ almost surely.
	Furthermore, almost surely and for all but finitely many $N \in \mathds N$ one has
	\begin{align}
	\dfrac{\Nlargest}{\llargesttilde} \le \sqrt{ 8 c_1(\eta)} \dfrac{\nu_N}{[\ln(\nu_N N)]^{1/2}} N^{3/5 + \delta}
	\end{align}
	for any $\delta > 0$ if $\eta < 1/6$ and $g_N \gg \nu_N N^{-1/6} [\ln (N)]^{-2}$.
\end{theorem}
\begin{proof}
	We define $\widehat g_N(\beta) := \nu_N N^{- \beta} [ \ln (N)]^{-2}$ for any $\beta > 0$.
	Let $\omega \in \Omega$ and $\widetilde N=\widetilde N(\omega) \in \mathds N$ such that the upper bound \eqref{equation largest interval} for the length of the largest interval with $\kappa = 5$, inequality $N^{1/3} \llargesttilde \widehat g_N(1/3) < (2 \max \{ \sqrt{2} c, \tilde c\})^{-2}$ with constants $c,\tilde c > 0$ from Lemma~\ref{lemma 1}, and inequality \eqref{asymptotic 1 proof M_j N_j as} of Theorem \ref{theorem upper bound ELS} hold for any $N \ge \widetilde N$. Then, with Theorem~\ref{theorem upper bound ELS}, Theorem~\ref{thm energy bound} ($\epsilon = 1/2$), equation \eqref{EGP scaling}, and inequality \eqref{bound EGP} we have
	\begin{align}
	c_1(\eta) \dfrac{\nu_N^2}{[\ln (N)]^2} & \ge \dfrac{1}{N} E_{\mathrm{LS}}^{\mathrm{QM}}(N, \Lambda_N, g_N) \\
	& \ge \dfrac{1}{N} E_0^{\mathrm{QM}}(\Nlargest, \llargesttilde, g_N) \\
	& \ge \dfrac{1}{N} E_0^{\mathrm{QM}} \left( \Nlargest, \llargesttilde, \min\left\{g_N, \widehat g_N(1/3) \right\} \right) \\
	& \ge \dfrac{1}{2} \dfrac{1}{N} E^{\mathrm{GP}} \left(\Nlargest, \llargesttilde, \min\left\{g_N, \widehat g_N(1/3) \right\} \right) \\
	& \ge  \dfrac{1}{4} \dfrac{1}{N} \dfrac{ (\Nlargest)^2}{\llargesttilde} \min\left\{g_N, \widehat g_N(1/3) \right\}
	\end{align}
	and therefore
	\begin{align}
	\Nlargest & \le \sqrt{4 c_1(\eta)} \dfrac{\nu_N}{\ln (N)} \left( \dfrac{\llargesttilde}{\min\left\{ g_N, \widehat g_N(1/3) \right\} } \right)^{1/2} N^{1/2}
	\end{align}
	for any $N \ge \widetilde N$. The first part of the statement then follows taking into account that almost surely $\llargesttilde \ge (1/2) \nu_N^{-1} \ln( \nu_N N)$ for all but finitely many $N \in \mathds N$, see Theorem~\ref{theorem lower bound largest interval}.
	
	Recall that $g_N \gg \nu_N N^{-1/6} [\ln(N)]^{-2}$ by assumption. We then define $\beta_0 := 1/3$ and $\beta_n :=  ( 1 + \beta_{n-1} )/6$ for any $n \in \mathds N$. We now show by induction that for any $n \in \mathds N$, $(\Nlargest)^{1/3} \llargesttilde \widehat g_N (\beta_n)$ converges to zero almost surely: Firstly, one has $\Nlargest \le N$ for any $N \in \mathds N$ and $N^{1/3} \llargesttilde \widehat g_N (\beta_0)$ converges to zero almost surely. Next, we assume that for arbitrary $n \in \mathds N$, $(\Nlargest)^{1/3} \llargesttilde \widehat g_N (\beta_n)$ converges to zero almost surely. Then for almost any $\omega \in \Omega$ there exists an $\widetilde N_n=\widetilde N_n(\omega) \in \mathds N$ such that $\widehat g_N(\beta_n) \le g_N$, $(\Nlargest)^{1/3} \llargesttilde \widehat g_N (\beta_n) < (2 \max \{ \sqrt{2} c, \tilde c\})^{-2}$, and inequality \eqref{asymptotic 1 proof M_j N_j as} of Theorem \ref{theorem upper bound ELS} holds for any $N \ge \widetilde N_n$. We therefore can conclude that
	\begin{align}
	\begin{split} \label{maximal particle density proof 1}
	c_1(\eta) \dfrac{\nu_N^2}{[\ln (N)]^2} & \ge \dfrac{1}{N} E_{\mathrm{LS}}^{\mathrm{QM}}(N, \Lambda_N, g_N) \\ &\ge \dfrac{1}{N} E_0^{\mathrm{QM}}(\Nlargest, \llargesttilde, g_N) \\
	& \ge \dfrac{1}{N} E_0^{\mathrm{QM}}(\Nlargest, \llargesttilde, \widehat g_N(\beta_n))\\ &\ge \dfrac{1}{2} \dfrac{1}{N} E^{\mathrm{GP}}(\Nlargest, \llargesttilde, \widehat g_N(\beta_n)) \ge \dfrac{1}{4} \dfrac{1}{N} \dfrac{ (\Nlargest)^2 \widehat g_N(\beta_n)}{\llargesttilde}
	\end{split}
	\end{align}
	or, equivalently,
	\begin{align}
	\Nlargest \le \sqrt{4 c_1(\eta)} ( \nu_N \llargesttilde )^{1/2} N^{(1 + \beta_n)/2}
	\end{align}
	for any $N \ge \widetilde N_n$. Hence, $(\Nlargest)^{1/3} \llargesttilde \widehat g_N (\beta_{n+1})$ converges to zero almost surely.
	
	Lastly, note that $(\beta_n)_{n \in \mathds N_0}$ converges to $1/5$. For arbitrary $\delta > 0$ we choose an $n \in \mathds N$ such that $\beta_n \le 1/5 + 2\delta$.  Hence, 
	\begin{align}
	\Nlargest \le \sqrt{4 c_1(\eta)} ( \nu_N \llargesttilde )^{1/2} N^{(1 + \beta_{n})/2} \le c_4(\eta) \nu_N N^{(1 + 1/5 + 2 \delta)/2} \llargesttilde [\ln (\nu_N N )]^{-1/2}
	\end{align}
	for all but finitely many $N \in \mathds N$ almost surely.
\end{proof}
Note that Theorem~\ref{TheoremParticleDensity} implies the following: For interactions $g_N \gg \nu_N N^{-1/6} [\ln(N)]^{-2}$, the particle density in the largest interval is almost surely bounded (actually converging to zero) in case of $[\ln(N)]^4/N^{1-2\eta} \ll \nuN \lesssim N^{-3/5 - \delta}$ for any $0 < \eta < 1/6$, $0 < \delta < 2/5 - 2\eta$ and it is asymptotically bounded by $\nu_N N^{3/5 + \delta}$ for any $\delta> 0$ in the case of $[\ln(N)]^4/N^{1-2\eta} \ll \nuN \lesssim 1$. In particular, note that for $[\ln(N)]^4/N^{1-2\eta} \ll \nuN \lesssim N^{-3/5 - \delta}$ the particle density in the largest interval diverges in the non-interacting model, i.e., if $g_N \equiv 0$. Hence, we conclude that the repulsive interaction between the particles is pivotal.

\appendix

%
\section{Notation}\label{Notation}
For two real-valued sequences $(a_N)_{N \in \mathds N}$, $(b_N)_{N \in \mathds N}$ with all elements positive and unequal to zero we write $a_N \sim b_N$ if there exist constants $c,C >0$ such that $c \le a_N/b_N \le C$ for all but finitely many $N \in \mathds N$. We also write $a_N \ll b_N$ if $a_N / b_N$ tends to zero.
We combine these two possibilities through writing $a_N \lesssim b_N$, meaning either $a_N \sim b_N$ or $a_N \ll b_N$.
Moreover, we also write $a_N \sim b_N$ in the case that $a_N = b_N = 0$ for all but finitely many $N \in \mathds N$ to simplify the notation.
\section{On the connection to the results of~\cite{SeiYngZag12} } \label{Appendix Seiringer gN nuN limit} 
We first note that, in contrast to the model discussed in this paper, in~\cite{SeiYngZag12} the unit interval is the fixed one-particle configuration space. However, by an appropriate scaling as discussed in~\cite[Sec.~4.4]{seiringer2012disordered2} and at the end of this section, the results can be translated into each other. 

The Hamiltonian in~\cite{SeiYngZag12} is formally given by 
\begin{equation}
H=\sum_{i=1}^{N}\left(-\partial^2_{z_i} +V_{\omega}(z_i)\right)+\frac{\gamma}{N}\sum_{i < j}\delta(z_i-z_j)
\end{equation}
where 
\begin{equation}
V_{\omega}(z):=\sigma\sum_{j}\delta(z-z^{\omega}_j)
\end{equation}
with $\gamma \geq 0$ the coupling parameter for the interaction among the particles, $\nu$ the density of scatterers $\{z^{\omega}_j\}$, $\sigma$ the strength of the scattering potential (note that $\sigma=\infty$ in our model), and $m$ the number of scatterers in the unit interval. In \cite[Theorem~2.2]{SeiYngZag12} which is subsequently used to prove BEC they established the estimate
\begin{align} \label{Seiringer Beweis BEC}
\left( 1 - \dfrac{N_0}{N} \right) \le (const.) \dfrac{e_0}{e_1 - e_0} N^{-1/3} \min \{ \gamma, \gamma^{1/2}\}
\end{align}
with $N$ the total number of particles and $N_0$ the number of particles occupying the minimizer of the Gross--Pitaevskii functional. In addition, \cite[Lemma~5.1]{SeiYngZag12} provides the lower bound $e_1 - e_0 \ge \eta \ln(1 + \pi \mathrm{e}^{-2\eta})$ with $\eta = \sqrt{\pi^2 + 3 m \sigma + 3 \gamma}$. 

Now, with $\gamma \ge 1$, $e_0 = E^{\text{GP}}(1,1,\gamma) \ge \gamma/2$, and $\ln(1 + x) \le x$ for $x > 0$ it follows
\begin{align}
\dfrac{e_0}{\eta \ln(1 + \pi \mathrm{e}^{-2\eta})} N^{-1/3} \min \{ \gamma, \gamma^{1/2}\} \gtrsim \dfrac{\gamma^{3/2}}{\eta} N^{-1/3} \mathrm{e}^{2\eta} \ .
\end{align}
In the case of $\gamma \ge m \sigma$ it is
\begin{align}
\dfrac{\gamma^{3/2}}{\eta} N^{-1/3} \mathrm{e}^{2\eta} \gtrsim \dfrac{\gamma^{3/2}}{\sqrt{\gamma}}N^{-1/3} \mathrm{e}^{2 \sqrt \gamma}\ .
\end{align}
If $\gamma \ge [\ln (N)]^2$, this converges to infinity and therefore~\eqref{Seiringer Beweis BEC} does not prove BEC. On the other hand, if $\gamma \le m \sigma$ then
\begin{align}
\dfrac{\gamma^{3/2}}{\eta} N^{-1/3}\mathrm{e}^{2 \eta} \gtrsim \dfrac{1}{\sqrt{m \sigma}} N^{-1/3}\mathrm{e}^{2 \sqrt{m \sigma}} \ ,
\end{align}
which, for $m \sigma \ge [\ln (N)]^2$, again converges to infinity. Hence, in order to establish BEC with the estimate~\eqref{Seiringer Beweis BEC} it must hold that $m \sigma \le [\ln (N)]^2$ and therefore $\gamma \le [\ln (N)]^2$. 

Furthermore, the assumptions in~\cite[Theorem~3.1, Lemma~3.2, Lemma~3.3]{SeiYngZag12} are such that $\nu \to \infty$, $\gamma \to \infty$, and $\gamma \gg \nu / (\ln \nu)^2$. Therefore
\begin{align}
[\ln (N)]^2 \ge \gamma \gg \nu / (\ln \nu)^2 \gtrsim \nu^{1 - \epsilon}
\end{align}
for any $\epsilon > 0$. This however implies that $\nu$ must grow slower than $[\ln (N)]^{3}$.

We are now in position to compare with our results: as described in~\cite[Sec.~4.4]{seiringer2012disordered2} the above translates to a density {$\nu_N \ll [\ln (N)]^{3} / N$ when working on the interval \linebreak $(-L_N/2 , L_N/2)$. Furthermore, since the condition $\nu \gg 1$ implies $\nu_N \gg 1/N$ one concludes that $1/N \ll \nu_N \ll \ln (N)^{3} / N$ is a necessary requirement in \cite{SeiYngZag12} to prove BEC. Also, $\gamma = L_N N \gN \le [\ln (N)]^2$  implies the requirement $\gN \lesssim [\ln (N)]^2 / N^2$. 

Hence, comparing with Theorem~\ref{Theorem generalized BEC}, Theorem~\ref{Theorem transition of condensation} and Theorem~\ref{TheoremParticleDensity} we see that we are able to allow for densities $\nu_N$ which converge to zero more slowly or even are constant.
\section{Miscellaneous results}
Let $(\hat \Lambda_N)_{N \in \mathds N}$ be an arbitrary sequence of intervals in $\mathds R$. For any $N \in \mathds N$, we define $\hat L_N := |\hat 
\Lambda_N|$ and $\kappa_N$ as the number of atoms (impurities) of the Poisson random measure 
with intensity $\nu > 0$ within the interval 
$\hat \Lambda_N$.

The following (large deviation type) lemma is needed for the proof of the subsequent Theorem~\ref{Theorem k LN to nu}. Note here that $1- \theta+ \theta\ln \theta > 0$ for $\theta \in (0,\infty) / \{1\}$. 
\begin{lemma} \label{wichtige Abschaetzung chi}
	Let $\nu > 0$ and $N \in \mathds N$ be given. Then for any $\theta \ge 1$
	\begin{align}
	\mathds{P}\left( \kappa_N \ge \theta \nu \hat L_N \right) \le \mathrm{e}^{- \nu \hat L_N (1 - \theta+ \theta \ln \theta)} \ ,
	\end{align}
	and for any $0 < \theta\le 1$
	\begin{align}
	\mathds{P}\left( \kappa_N \le \theta \nu \hat L_N \right) \le \mathrm{e}^{- \nu \hat L_N (1 - \theta+ \theta\ln \theta)} \ .
	\end{align}
\end{lemma}
\begin{proof}
	For $\theta\ge 1$ we have
	\begin{align*}
	\mathds{P}\left( \kappa_N \ge \theta \nu \hat L_N \right) & =  \sum\limits_{m \ge \theta \nu \hat L_N} \mathds{P} ( \kappa_N = m ) = \sum\limits_{m \ge \theta \nu \hat L_N} \mathrm{e}^{- \nu \hat L_N} \dfrac{(\nu \hat L_N)^m}{m!} \\
	& \le \sum\limits_{m \ge \theta \nu \hat L_N} \mathrm{e}^{- \nu \hat L_N} \dfrac{( \nu \hat L_N )^m}{m!} \theta^{m - \theta \nu \hat L_N} \le \mathrm{e}^{- \nu \hat L_N \left( 1 - \theta + \theta\ln \theta\right)} \ .
	\end{align*}
	On the other hand, for $0 < \theta\le 1$,
	\begin{align*}
	\mathds{P}\left( \kappa_N \le \theta \nu \hat L_N \right) & =  \sum\limits_{m \le \theta \nu \hat L_N} \mathds{P} \left( \kappa_N = m \right) = \sum\limits_{m \le \theta \nu \hat L_N} \mathrm{e}^{- \nu \hat L_N} \dfrac{(\nu \hat L_N)^m}{m!} \\
	& \le \sum\limits_{m \le \theta \nu \hat L_N} \mathrm{e}^{- \nu \hat L_N} \dfrac{( \nu \hat L_N )^m}{m!} \theta^{m - \theta \nu \hat L_N} \le \mathrm{e}^{- \nu \hat L_N \left( 1 - \theta + \theta\ln \theta\right)}
	\end{align*}
\end{proof}
	\begin{theorem} \label{Theorem k LN to nu}
		Let $(\hat \Lambda_N)_{N \in \mathds N}$ with $\hat L_N \gg \ln(N)$ be given. Then, for any $\epsilon > 0$ and for almost any $\omega \in \Omega$ there exists an $\widetilde N = \widetilde N(\epsilon,\omega) \in \mathds N$ such that for any $N \ge \widetilde N$ we have
		\begin{align}
		(1 - \epsilon) \nu \, \hat L_N \le \kappa_N \le (1 + \epsilon) \nu \, \hat L_N \ . 
		\end{align}
		In particular, almost surely $\lim_{N \to \infty} k_N  / L_N = \nu$ and $\lim_{N \to \infty} \tilde k_N / (\nu_N L_N) = 1$ in case of $\ln (N) / N \ll \nu_N \lesssim 1$.
	\end{theorem}
	\begin{proof}
		Let $\epsilon > 0$ be given. Then, with Lemma~\ref{wichtige Abschaetzung chi} we obtain
		\begin{align*}
		\sum\limits_{N =1}^{\infty} \mathds{P}\Big( \kappa_N \le (1 - \epsilon) \nu \hat L_N \Big) < \infty \quad \text{ and }\quad \sum\limits_{N =1}^{\infty} \mathds{P}\Big( \kappa_N \ge (1 + \epsilon) \nu \hat L_N \Big)< \infty \ .
		\end{align*}
		Hence, the first part of the statement follows with the Borel--Cantelli lemma.  
		
		Consequently, for any $\epsilon > 0$ and almost any $\omega \in \Omega$,
		\begin{align*}
		\liminf\limits_{N \to \infty} \dfrac{\kappa_N}{\hat L_N} \ge (1 - \epsilon)\nu \quad \text{ and } \quad \limsup\limits_{N \to \infty} \dfrac{\kappa_N}{\hat L_N} \le (1 + \epsilon)\nu \ .
		\end{align*}
		Setting $\hat \Lambda_N = \Lambda_N$, $N \in \mathds N$, we conclude that almost surely $\lim_{N \to \infty} k_N/L_N = \nu$. Furthermore, recalling that $\nu_N = s_N \nu$ we obtain $\lim_{N \to \infty} \tilde k_N/ (\nu_N L_N) = 1$ almost surely by setting $\hat \Lambda_N = s_N \Lambda_N$, $N \in \mathds N$. Note here that the assumption $\ln(N)/N \ll \nu_N \lesssim 1$ implies $\hat L_N \gg \ln(N)$.
	\end{proof}
	For a Poisson random measure with intensity $\nu > 0$ we define $\hat l_j := |(x_j(\omega), x_{j+1}(\omega))|$ for any $j \in \mathds Z$ with $\{ x_j (\omega) : j \in \mathds Z \}$ the strictly increasing sequence of the atoms of the Poisson random measure, see Section~\ref{TheModel}. Note that $\{\hat l_j : j \in \mathds Z \backslash \{0\} \}$ are independent and identically distributed random variables with common density $\nu \mathrm{e}^{-\nu l}$ \cite[Ch. 4]{kingman1993poisson}. We also define the set $J_k :=  \{ - k, -k + 1, \ldots, k - 1, k \} \backslash \{ 0 \}$ for any $k \in \mathds N$.
	
	\begin{lemma} \label{Lemma hat lj equal lj}
		Assume that $\ln (N) / N \ll \nu_N\lesssim 1$ holds. Then, for any $0 < \epsilon < 1$ and almost any $\omega \in \Omega$ there exists an $\widetilde N = \widetilde N(\epsilon,\omega) \in \mathds N$ such that for any $N \ge \widetilde N$ one has
		$\tilde l_j = s_N^{-1} \hat l_j$ for any $j \in J_{\lceil (1 - \epsilon) \nu_N L_N/2 \rceil}$,
		$\tilde l_j \le s_N^{-1} \hat l_j$ for any $j \in J_{\lfloor \nu_N L_N \rfloor}$ and $\tilde l_j = 0$ for any $j \in \mathds Z \backslash ( J_{\lfloor \nu_N L_N \rfloor} \cup \{0\})$. 
	\end{lemma}
	\begin{proof}
		Let $0 < \epsilon < 1$ and $0 < \epsilon^{\prime} < \epsilon$ be given. For any $N \in \mathds N$ we divide the window $(- s_N L_N/2, s_N L_N/2)$ into $(- s_N L_N/2,0]$ and $[0,- s_N L_N/2)$.
		Due to Theorem~\ref{Theorem k LN to nu} there exists a set $\tilde \Omega \subset \Omega$ with $\mathds{P}(\widetilde \Omega) = 1$ and the following property: For any $\omega \in \widetilde \Omega$ there exists an $\widetilde N\in \mathds N$ such that for any $N \ge \widetilde N$ one has $(1/2)(1 - \epsilon^{\prime})\nu_N L_N \le \kappa_N^{(1)}, \kappa_N^{(2)} \le (1/2)(1 + \epsilon^{\prime}) \nu_N L_N$ with $\kappa_N^{(1)}$ and $\kappa_N^{(2)}$ denoting the number of atoms within $(- s_N L_N/2,0]$ and $[0, s_N L_N/2)$, respectively. 
		
		The statement of the lemma now follows since $s_N^{-1} \hat l_j = s_N^{-1} |(x_j(\omega), x_{j+1}(\omega))|$ and $\tilde l_j = s_N^{-1} |(x_j(\omega), x_{j+1}(\omega)) \cap s_N \Lambda_N|$ for any $j \in \mathds Z$. Note that we divide the window $s_N \Lambda_N = (-s_N L_N/2, s_N L_N/2)$ into the two intervals in order to ensure that, for $N \ge \widetilde N$, \textit{both} intervals with associated lengths $\tilde l_{-\lceil (1 - \epsilon) \nu_N L_N/2 \rceil}$ and $\tilde l_{\lceil (1 - \epsilon) \nu_N L_N/2 \rceil}$ are entirely within the window $s_N \Lambda_N$ and that $\tilde l_{j} = 0$ for any $j \le - \lfloor \nu_N L_N \rfloor - 1$ \textit{and} for any $j \ge \lfloor \nu_N L_N \rfloor + 1$.
	\end{proof}
	\begin{theorem}	\label{thm ConvergenceParticleNumber almost surely}
		Assume that $\ln(N) / N \ll \nu_N \lesssim 1$ holds. Furthermore, let $(g_N)_{N \in \mathds N}$ be such that $\lim_{N \to \infty} \nu_N \lambdaN = 0$ and $\lambda_N \gg \ln (N) \ln (\nu_N N) (\nu_N N)^{-1/2}$ with $(\lambda_N)_{N \in \mathds N}$ as in~\eqref{Zusammenhang g lambda}. Then 
		\begin{equation}\label{ConvergenceParticleNumber}
		\liminf\limits_{N \to \infty} \dfrac{1}{N} \sum_{j \in \mathds 
Z} M_j \ge \dfrac{2}{9} \quad \text{ and } \quad \limsup\limits_{N \to \infty} 
\dfrac{1}{N} \sum_{j \in \mathds Z} M_j \le 6
		\end{equation}
		almost surely with $M_j$ as in~\eqref{Besetzungszahlen}.
	\end{theorem}
	\begin{proof}
		Let $F^{\nu}_{k}(l) := (2 k)^{-1} \sum_{j \in J_k} \mathds 1_{\hat l_j \le l}$ be the empirical distribution function with respect to the random variables $\{ \hat l_j : j \in \mathds Z \backslash \{0\} \}$. Then, for any $l \in \mathds R$, there is a unique measure $F^{\nu}_k(\mathrm{d}\ell)$ defined by $\int_{-\infty}^l F^{\nu}_k(\mathrm{d}\ell) := F^{\nu}_k(l)$ and we also set $F^{\nu}(l) := \int_{0}^{l} \mathds 1_{(0,\infty)}(l) \nu \mathrm{e}^{- \nu x} \, \mathrm{d}x \,$.
		
		Then $\mathds{P} ( \sup_{l \in \mathds R} \left| F^{\nu}_{k}(l) - F^{\nu}(l) \right| > \epsilon ) \le 2 \mathrm{e}^{-2 k \epsilon^2}$ for any $\epsilon > 0$ and $k \in \mathds N$ due to Dvoretzky--Kiefer--Wolfowitz inequality \cite{dvoretzky1956asymptotic, massart1990tight}.
		Therefore, for any $\epsilon > 0$ and with $K \in \mathds N$ such that $\ln(K) \ge \epsilon^{-2}$ we obtain
		\begin{align}
		& \sum\limits_{k = K}^{\infty} \mathds{P} \left( \dfrac{k^{1/2}}{ \ln(k)} \sup\limits_{l \in \mathds R} \left| F^{\nu}_{k}(l) - F^{\nu}(l) \right| > \epsilon \right) \le 2 \sum\limits_{k=K}^{\infty} \mathrm{e}^{-2 \ln(k)} <  \infty \ .
		\end{align}
		Hence,
		\begin{align} \label{convergence empirical distribution}
		\lim\limits_{k \to \infty} \dfrac{k^{1/2}}{\ln(k)} \sup\limits_{l \in \mathds R} \left| F^{\nu}_{k}(l) - F^{\nu}(l) \right| = 0
		\end{align}
		almost surely by the Borel--Cantelli lemma. 
		
		Note that one has $\muN \ge \pi^2 / [5 \nu_N^{-1} \ln (N)]^{2}$ for all but finitely many $N \in \mathds{N}$ since otherwise $\lambda_N \le N^{-5}$ which then contradicted $\lambda_N \gg \ln(N) \ln(\nu_N N)(\nu_N N)^{-1/2}$. Moreover, by Lemma~\ref{remark largest interval}, $\llargest \le \kappa \nu^{-1} \ln(N)$ for any $\kappa > 4$ and therefore $F^{\nu}_{\lceil \nu_N L_N/4 \rceil}(l) = F^{\nu}_{\lfloor \nu_N L_N \rfloor}(l) = 1$ for $l \ge 5 \nu^{-1} \ln (N)$ for all but finitely many $N \in \mathds N$ almost surely.
		
		Therefore, almost surely with $F^{\nu}(\mathrm{d}\ell)=\nu \mathrm{e}^{-\nu \ell}\mathrm{d}\ell$,
		  \begin{equation}\begin{split}
		 C_N:&=\left| \int\limits_{\pi s_N / \sqrt{\muN}}^{\infty} \dfrac{s_N}{\ell \gN} \left( \muN \dfrac{\ell^2}{s_N^2} - \pi^2 \right) \, \left(F^{\nu}_{\lceil \nu_N L_N  / 4 \rceil}(\mathrm{d}\ell) - F^{\nu}(\mathrm{d}\ell)\right) \right| \\
		\quad &\leq \dfrac{\muN}{g_N s_N} \left| \int\limits_{\pi s_N / \sqrt{\muN}}^{5 \nu^{-1} \ln (N)} \ell \, \left(F^{\nu}_{\lceil \nu_N L_N  / 4 \rceil}(\mathrm{d}\ell) - F^{\nu}(\mathrm{d}\ell)\right) \right| + \dfrac{\muN}{g_N s_N} \int\limits_{5 \nu^{-1} \ln (N)}^{\infty}   \ell \, F^{\nu}(\mathrm{d}\ell)\, \\
		 & \qquad + \, \pi^2 \dfrac{s_N}{g_N} \left| \int\limits_{\pi s_N / \sqrt{\muN}}^{5 \nu^{-1} \ln (N)} \ell^{-1} \, \left(F^{\nu}_{\lceil \nu_N L_N  / 4 \rceil}(\mathrm{d}\ell) - F^{\nu}(\mathrm{d}\ell)\right)  \right|\,\\
		 & \qquad + \, \pi^2 \dfrac{s_N}{g_N} \int\limits_{5 \nu^{-1} \ln (N)}^{\infty} \ell^{-1} \, F^{\nu}(\mathrm{d}\ell)\ .
		 \end{split}
		 \end{equation}
		 Since $F_{\lceil \nu_N L_N / 4 \rceil}^{\nu}(l) = 1$ for any $l \ge 5 \nu^{-1} \ln(N)$ we obtain, with an integration by parts in the first and third term,
		 \begin{equation}\begin{split}
		 C_N &\leq \dfrac{\muN}{g_N s_N} \left( 5 \nu^{-1} \ln (N) \right) \left| F^{\nu}_{\lceil \nu_N L_N  / 4 \rceil}(5 \nu^{-1} \ln (N)) - F^{\nu}(5 \nu^{-1} \ln (N)) \right|\, \\
		 & \qquad + \, \dfrac{\muN}{g_N s_N} \dfrac{\pi s_N}{\sqrt{\muN}} \left| F^{\nu}_{\lceil \nu_N L_N  / 4 \rceil}(\pi s_N / \sqrt{\muN}) - F^{\nu}(\pi s_N / \sqrt{\muN}) \right| \,\\
		 & \qquad + \, \dfrac{\muN}{g_N s_N} \left| \int\limits_{\pi s_N / \sqrt{\muN} }^{5 \nu^{-1} \ln (N)} \left(F^{\nu}_{\lceil \nu_N L_N  / 4 \rceil}(\ell)- F^{\nu}(\ell)\right) \, \mathrm{d}\ell \right| + \, \dfrac{\muN}{g_N s_N} \int\limits_{5 \nu^{-1} \ln (N)}^{\infty}   \ell \, \nu \mathrm{e}^{-\nu \ell} \mathrm{d}\ell \, \\
		 & \qquad \, \\
		 & \qquad + \,  \pi^2 \dfrac{s_N}{g_N} \left( 5 \nu^{-1} \ln (N) \right)^{-1} \left| F^{\nu}_{\lceil \nu_N L_N  / 4 \rceil}(5 \nu^{-1} \ln (N)) - F^{\nu}(5 \nu^{-1} \ln (N)) \right| \, \\
		 & \qquad + \, \pi^2 \dfrac{s_N}{g_N} \dfrac{\sqrt{\muN}}{\pi s_N} \left| F^{\nu}_{\lceil \nu_N L_N  / 4 \rceil}(\pi s_N / \sqrt{\muN}) -  F^{\nu}(\pi s_N / \sqrt{\muN}) \right| \, \\
		 & \qquad + \, \pi^2 \dfrac{s_N}{g_N} \left| \int\limits_{\pi s_N / \sqrt{\muN}}^{5 \nu^{-1} \ln (N)} \ell^{-2}\left(F^{\nu}_{\lceil \nu_N L_N  / 4 \rceil}(\ell) \, \mathrm{d}\ell - F^{\nu}(\ell) \right) \, \mathrm{d}\ell  \right| \, \\
		 & \qquad + \, \pi^2 \dfrac{s_N}{g_N} \int\limits_{5 \nu^{-1} \ln (N)}^{\infty} \ell^{-1} \nu \mathrm{e}^{-\nu \ell} \, \mathrm{d}\ell \ .
		  \end{split}
		  \end{equation}
		 Calculating further we obtain 
		 \begin{equation}\begin{split}
		 &\text{}\\
		  C_N &\leq 3 \dfrac{5 \ln (N)}{\nu_N} \dfrac{\muN}{g_N} \sup\limits_{l \in \mathds{R}} |F^{\nu}_{\lceil \nu_N L_N  / 4 \rceil}(l) - F^{\nu}(l)  | + \dfrac{\muN}{g_N} \dfrac{1}{\nu_N} \dfrac{5 \ln (N) + 1}{N^{5}} \, \\
		 & \quad + \, 3 \dfrac{5 \ln (N)}{\nu_N}  \dfrac{\muN}{g_N} \sup\limits_{l \in \mathds{R}} |F^{\nu}_{\lceil \nu_N L_N  / 4 \rceil}(l) - F^{\nu}(l)  | + \dfrac{\muN}{g_N} \dfrac{1}{\nu_N} \dfrac{5 \ln (N)}{N^{5}} \\
		 \end{split}
		 \end{equation}
		 and since $\mu_N \ge \pi^2 / [5 \nu_N^{-1} \ln(N)]^2$ we get
		 \begin{equation}\begin{split}
		  C_N &\leq 6 \cdot 5 \cdot \rho \dfrac{1}{\nu_N} 
\dfrac{\ln (N)}{\lambda_N} \sup\limits_{l \in \mathds{R}} |F^{\nu}_{\lceil \nu_N 
L_N  / 4 \rceil}(l) - F^{\nu}(l)  | + \rho \dfrac{10 \ln (N) + 
1}{\nu_N \lambda_N} \dfrac{1}{N^5}
		\end{split}
		  \end{equation}
		  with~\eqref{Zusammenhang g lambda} for all sufficiently large $N \in \mathds N$. Hence, with \eqref{convergence empirical distribution} and due to $\lambda_N \gg \ln (N) \ln (\nu_N N) (\nu_N N)^{-1/2}$ we conclude that $\lim_{N \rightarrow \infty}C_N=0$.
		  
		  Since
		  \begin{align}
		  \nu_N \int\limits_{\pi s_N / \sqrt{\muN}}^{\infty} \dfrac{s_N}{\ell \gN} \left( \muN \dfrac{\ell^2}{s_N^2} - \pi^2 \right) \, F^{\nu}(\mathrm{d}\ell) & = \nu_N \E^{\nu} \left[ \dfrac{s_N}{(\cdot) g_N} \left[ \mu_N \dfrac{(\cdot)^2}{s_N^2} - \pi^2 \right]_+ \right] \\
		  & \ge \nu_N \E^{\nu}\left[ N_{\gN, \muN}(s_N^{-1} (\cdot)) \right] \geq 
\dfrac{2}{3} \rho
		  \end{align}
		  for any $N \in \mathds N$, see \eqref{inequality for Ngmu(l)} and \eqref{inequality for E NgNmuN}. Hence we conclude that almost surely
		  \begin{align} \label{int fdl ge rho 1}
		  \begin{split}
		  & \liminf \limits_{N \to \infty} \left(\nu_N \int\limits_{\pi s_N / \sqrt{\muN}}^{\infty} \dfrac{s_N}{\ell \gN} \left( \muN \dfrac{\ell^2}{s_N^2} - \pi^2 \right) \, F^{\nu}_{\lceil \nu_N L_N  / 4 \rceil}(\mathrm{d}\ell)\right) \\
		  &= \,  \liminf\limits_{N \to \infty} \left(\nu_N \int\limits_{\pi 
s_N / \sqrt{\muN}}^{\infty} \dfrac{s_N}{\ell \gN} \left( \muN \dfrac{\ell^2}{s_N^2} - 
\pi^2 \right) \, F^{\nu}(\mathrm{d}\ell)\right) \ge \dfrac{2}{3} \rho \ .
		  \end{split}
		  \end{align}
		  Also, repeating the arguments from above we one can show that 
almost surely
		  \begin{align}
		  \lim\limits_{N \to \infty} & \left| \nu_N \int\limits_{\pi s_N / \sqrt{\muN}}^{\infty} \dfrac{s_N}{\ell \gN} \left( \muN \dfrac{\ell^2}{s_N^2} - \pi^2 \right) \, \left(F^{\nu}_{\lfloor \nu_N L_N \rfloor}(\mathrm{d}\ell) \, - F^{\nu}(\mathrm{d}\ell)\right) \right| = 0 \ ,
		  \end{align}
		  \begin{align}
		  \nu_N \int\limits_{\pi s_N / \sqrt{\muN}}^{\infty} 
\dfrac{s_N}{\ell \gN} \left( \muN \dfrac{\ell^2}{s_N^2} - \pi^2 \right) \, F^{\nu}(\mathrm{d}\ell) 
& \le \dfrac{3}{2} \nu_N \E^{\nu}\left[ N_{\gN, \muN}(s_N^{-1} (\cdot)) \right] \leq
3 \rho
		  \end{align}
		  for any $N \in \mathds N$, see \eqref{inequality for Ngmu(l)} and \eqref{inequality for E NgNmuN}. Consequently 
		  \begin{align} \label{int fdl ge rho 2}
		  \begin{split}
		  & \limsup\limits_{N \to \infty} \left(\nu_N \int\limits_{\pi s_N / \sqrt{\muN}}^{\infty} \dfrac{s_N}{\ell \gN} \left( \muN \dfrac{\ell^2}{s_N^2} - \pi^2 \right) \, F^{\nu}_{\lfloor \nu_N L_N \rfloor}(\mathrm{d}\ell)\right) \\
		  &= \,  \limsup\limits_{N \to \infty} \left(\nu_N \int\limits_{\pi 
s_N / \sqrt{\muN}}^{\infty} \dfrac{s_N}{\ell \gN} \left( \muN \dfrac{\ell^2}{s_N^2} - 
\pi^2 \right) \, F^{\nu}(\mathrm{d}\ell)\right) \le 3 \rho \ .
		  \end{split}
		  \end{align}

		  	Due to Lemma~\ref{Lemma hat lj equal lj} (with $\epsilon = 1/2$), there exists a set $\widetilde \Omega \subset \Omega$ with $\mathds{P}(\widetilde \Omega) = 1$ and the following property: For any $\omega \in \widetilde \Omega$ there exists an $\widehat N(\omega) \in \mathds N$ such that for any $N \ge \widehat N(\omega)$ we have $\tilde l_j = s_N^{-1} \hat l_j$ for any $j \in J_{\lceil \nu_N L_N/4 \rceil}$, $\tilde l_j \le s_N^{-1} \hat l_j$ for any $j \in J_{\lfloor \nu_N L_N \rfloor}$, and $\tilde l_j = 0$ for any $j \in \mathds Z \backslash ( J_{\lfloor \nu_N L_N \rfloor} \cup \{ 0 \})$. 
		  	Consequently we obtain ($M_0 = 0$), for any $N \ge \widehat N(\omega)$,
		  \begin{align}\label{hat lj subset lj}
		  \begin{split}
		  \sum\limits_{j \in J_{ \lceil \nu_N L_N / 4 \rceil}} \dfrac{s_N}{\hat l_j g_N} \left[ \mu_N \dfrac{\hat l_j^2}{s_N^2} - \pi^2 \right]_+ & \le \sum\limits_{j \in \mathds Z: M_j \ge 1} \dfrac{1}{\tilde l_j g_N} \left[ \mu_N \tilde l_j^2 - \pi^2 \right]_+ \\
		  & \le \sum\limits_{j \in J_{\lfloor \nu_N L_N \rfloor}} \dfrac{s_N}{\hat l_j g_N} \left[ \mu_N \dfrac{\hat l_j^2}{s_N^2} - \pi^2 \right]_+\ .
		  \end{split}
		  \end{align}

		  Lastly, due to \eqref{Besetzungszahlen}, \eqref{inequality for Ngmu(l)}, \eqref{hat lj subset lj}, \eqref{int fdl ge rho 1}, \eqref{int fdl ge rho 2}, and $M_0 = 0$, we have almost surely 
		  \begin{align*}
		  \liminf\limits_{N \to \infty} \dfrac{\nu_N}{ \nu_N L_N / 2 } \sum\limits_{j \in \mathds Z} M_j  & \ge \liminf\limits_{N \to \infty} \dfrac{\nu_N}{ \nu_N L_N / 2 } \sum\limits_{j \in \mathds Z: M_j \ge 1} N_{g_N, \mu_N}(\tilde l_j) \\
		  & \ge \dfrac{2}{3} \liminf\limits_{N \to \infty} \dfrac{\nu_N}{ \nu_N L_N / 2 } \sum\limits_{j \in \mathds Z: M_j \ge 1} \dfrac{1}{\tilde l_j g_N} \left[ \mu_N \tilde l_j^2 - \pi^2 \right]_+\\
		  & \ge \dfrac{2}{3} \liminf\limits_{N \to \infty} \dfrac{\nu_N}{ \nu_N L_N / 2 } \sum\limits_{j \in J_{\lceil \nu_N L_N / 4 \rceil}} \dfrac{s_N}{\hat l_j g_N} \left[ \mu_N \dfrac{\hat l_j^2}{s_N^2} - \pi^2 \right]_+ \\
		  & \geq \dfrac{2}{3} \liminf\limits_{N \to \infty} \nu_N \int\limits_{\pi s_N / \sqrt{\muN}}^{\infty} \dfrac{s_N}{l \gN} \left( \muN \dfrac{l^2}{s_N^2} - \pi^2 \right) \, F^{\nu}_{\lceil \nu_N L_N  / 4 \rceil}(\mathrm{d}\ell) \\
		  & \ge \left( \dfrac{2}{3} \right)^2 \rho\ .
		  \end{align*}
		  Similarly,
		  \begin{align*}
		  \limsup\limits_{N \to \infty} \dfrac{\nu_N}{2 \nu_N L_N} \sum\limits_{j \in \mathds Z} M_j  \le &\limsup\limits_{N \to \infty} \dfrac{\nu_N}{2 \nu_N L_N} \sum\limits_{j \in \mathds Z: M_j \ge 1} \dfrac{1}{\tilde l_j g_N} \left[ \mu_N \tilde l_j^2 - \pi^2 \right]_+ \\
		  & \quad + \, \lim\limits_{N \to \infty} \dfrac{\nu_N}{2 \nu_N L_N} \sum\limits_{j \in \mathds Z : M_j \ge 1} 1  \\
		   \le &\limsup\limits_{N \to \infty} \dfrac{\nu_N}{2 \nu_N L_N} \sum\limits_{j \in J_{\lfloor \nu_N L_N \rfloor}} \dfrac{s_N}{\hat l_j g_N} \left[ \mu_N \dfrac{\hat l_j^2}{s_N^2} - \pi^2 \right]_+ \\
		  & \quad + \, \lim\limits_{N \to \infty} \dfrac{\nu_N}{2 \nu_N L_N} \sum\limits_{j \in \mathds Z : M_j \ge 1} 1  \\
		   \le &\limsup\limits_{N \to \infty} \nu_N \int\limits_{\pi s_N / \sqrt{\muN}}^{\infty} \dfrac{s_N}{l \gN} \left( \muN \dfrac{l^2}{s_N^2} - \pi^2 \right) \, F^{\nu}_{\lfloor \nu_N L_N \rfloor}(\mathrm{d}\ell) \\
		  & \quad + \, \lim\limits_{N \to \infty} \nu_N \int\limits_{\pi 
s_N / \sqrt{\muN}}^{\infty} F^{\nu}_{\lfloor \nu_N L_N \rfloor}(\mathrm{d}\ell) \le 
3 \rho \ ,
		  \end{align*}
		  since, due to \eqref{convergence empirical distribution},~\eqref{definition of lambda}, and
our assumptions, 
		  \begin{align} \begin{split}
		  \lim\limits_{N \to \infty} \nu_N \int\limits_{\pi s_N / \sqrt{\muN}}^{\infty} F^{\nu}_{\lfloor \nu_N L_N \rfloor}(\mathrm{d}\ell) & = \lim\limits_{N \to \infty} \nu_N \big[ 1 - F^{\nu}_{\lfloor \nu_N L_N \rfloor} ( \pi s_N / \sqrt{\muN} ) \big] \\
		  & = \lim\limits_{N \to \infty} \nu_N \big[ 1 - F^{\nu} ( \pi s_N / \sqrt{\muN} ) \big] = \lim\limits_{N \to \infty}  \nu_N \lambdaN = 0 \ .
		  \end{split}
		  \end{align}
		  We therefore obtain almost surely
		  \begin{align}
		  & \liminf\limits_{N \to \infty} \dfrac{1}{N} \sum\limits_{j 
\in \mathds Z} M_j = \dfrac{1}{2 \rho} \left( \liminf\limits_{N \to \infty} 
\dfrac{\nu_N}{ \nu_N L_N / 2 } \sum\limits_{j \in \mathds Z} M_j \right) \ge 
\dfrac{2}{9} \ ,
		  \end{align}
		  and
		  \begin{align}
		  & \limsup\limits_{N \to \infty} \dfrac{1}{N} \sum\limits_{j 
\in \mathds Z} M_j \le 6 \ .
		  \end{align}
		\end{proof}

		 	\begin{theorem} \label{theorem lower bound largest interval}
		 		Let $1/N \ll \nu_N \lesssim 1$ be given. Then for any $0 < \epsilon <1 $ and for almost any $\omega \in \Omega$ there exists an $\widetilde N=\widetilde N(\epsilon,\omega) \in \mathds N$ such that for any $N \ge \widetilde N$ we have
		 		\begin{align}
		 		\llargesttilde > \nu_N^{-1} \Big\{ \ln( L_{\lfloor s_N N \rfloor}) - (1 + \epsilon) \ln[\ln (L_{\lfloor s_N N \rfloor})] \Big\} \ .
		 		\end{align}
		 	\end{theorem}
		 	\begin{proof}
		 		Let $0 < \epsilon <1 $ be given. Then, for almost any $\omega \in \Omega$ there exists an $\widehat N_1=\widehat N_1(\epsilon,\omega) \in \mathds N$ such that
		 		\begin{align}
		 		\big\{ \hat l_j : j \in J_{\lceil (1-\epsilon) \nu L_N / 2 \rceil} \big\} \subsetneq \big\{ l_j :  j \in \mathds Z \backslash \{0\} \big\}\backslash \{0\}
		 		\end{align}
		 		for any $N \ge \widehat N_1$ by Lemma~\ref{Lemma hat lj equal lj}. Moreover, since $\{ \hat l_j : j \in \mathds Z\}$ are mutually independent exponentially distributed random variables one obtains
		 		\begin{align}
		 		& \mathds{P} \left( \max \left\{ \hat l_j  : j \in J_{\lceil (1-\epsilon) \nu L_N / 2 \rceil} \right\} \le \nu^{-1} \Big\{ \ln(L_N) - ( 1 + \epsilon) \ln [ \ln (L_N) ] \Big\} \right)\\
		 		\le \, & \left( 1 - \dfrac{[ \ln(L_N)]^{1 + \epsilon}}{L_N} \right)^{2 \lceil (1 - \epsilon) \nu L_N/ 2 \rceil} \ .
		 		\end{align}
		 		Moreover, since $\ln(1 - x) \le -x$ for any $0 < x < 1$ we have
		 		\begin{align}
		 		2 \left\lceil \dfrac{(1 - \epsilon) \nu L_N}{2} \right\rceil \cdot \ln \left[ \left( 1 - \dfrac{[\ln(L_N)]^{1 + \epsilon}}{L_N} \right) \right] & \le - (1 - \epsilon) \nu [ \ln(L_N)]^{1 + \epsilon} \le - 2 \ln(N) 
		 		\end{align}
		 		for all but finitely many $N \in \mathds N$ and therefore
		 		\begin{align}
		 		\sum\limits_{N = 1}^{\infty} \left( 1 - \dfrac{[ \ln(L_N)]^{1 + \epsilon}}{L_N} \right)^{2 \lceil (1 - \epsilon) \nu L_N \rceil / 2} < \infty \ .
		 		\end{align}
		 		Hence, with Borel--Cantelli's lemma there exists a set $\widetilde \Omega \subset \Omega$ with $\mathds{P}(\widetilde \Omega) = 1$ such that for any $\omega \in \widetilde \Omega$ there is an $\widehat N_2=\widehat N_2(\epsilon,\omega) \in \mathds{N}$ with
		 		\begin{align}
		 		\llargest \ge  \max \left\{ \hat l_j  : j \in J_{\lceil (1-\epsilon) \nu L_N / 2 \rceil} \right\}  > \nu^{-1} \Big\{ \ln( L_N) - (1 + \epsilon) \ln [ \ln(L_N)] \Big\}
		 		\end{align}
		 		for any $N \ge \widehat N_2$.
		 		
		 		Finally, for any $\omega \in \widetilde \Omega$ we define $\widetilde N(\epsilon,\omega)$ such that $\lfloor s_N N \rfloor \ge \widehat N_2(\epsilon,\omega)$ for any $N \ge \widetilde N(\epsilon,\omega)$. Hence, $\ell^{>}_{\lfloor s_N N \rfloor}$ denoting the length of the largest interval within the window $\Lambda_{\lfloor s_N N \rfloor}$,
		 		\begin{align}
		 		s_{N}^{-1}\ell^{>}_{\lfloor s_N N \rfloor}  > \nu_N^{-1} \Big\{ \ln( L_{\lfloor s_N N \rfloor}) - (1 + \epsilon) \ln \big[ \ln (L_{\lfloor s_N N \rfloor}) \big] \Big\} 
		 		\end{align}
		 		for any $N \ge \widetilde N(\epsilon,\omega)$. The statement then follows since $\llargesttilde \geq s_{N}^{-1}\ell^{>}_{\lfloor s_N N \rfloor}$.
		 	\end{proof}

		 	\begin{theorem} \label{theorem C3}
		 		 For any $0<\eta'<2$ 
there exists an $\widetilde N(\eta') \in \mathds N$ such that for any $N \ge 
\widetilde N(\eta')$ one has
		 		\begin{align}
		 		\mathds{P} \Big( \llargest > \nu^{-1} \left[ \ln(L_N) + \ln \left( C_1 \right) \right] \Big) > 1 - \dfrac{1}{2}\eta'
		 		\end{align}
		 		with $C_1 := - \nu / [4 \ln(\eta' / 2)] > 0$.
		 	\end{theorem}
		 	\begin{proof}
		 	According to 
Lemma~\ref{Lemma hat lj equal lj} (with $\epsilon = 1/2$), for almost any 
$\omega \in \Omega$ there exists an $\widehat N=\widehat N(\omega) \in \mathds N$ such that for any 
$N \ge \widehat N$ one has
		 		\begin{align*}
		 		\big\{ \hat l_j : j \in J_{\lceil \nu L_N/4 \rceil} \big\} \subsetneq \big\{ l_j : j \in \mathds Z \backslash \{0\} \big\} \backslash \{0\}\ .
		 		\end{align*}
		 		Hence, $[\max \{ \hat l_j : j \in J_{\lceil \nu L_N/4 \rceil} \} - \llargest]_+$ converges to zero almost surely and consequently, for any $\eta > 0$,
		 		\begin{align*}
		 		 \lim\limits_{N \to \infty} &\mathds{P}\left( \llargest \le \max \big\{ \hat l_j : j \in J_{\lceil \nu L_N/4 \rceil} \big\} - \eta  \right) \\
		 		\le \, & \lim\limits_{N \to \infty} \mathds{P}\left( \left[ \max \big\{ \hat l_j : j \in J_{\lceil \nu L_N/4 \rceil} \big\} - \llargest \right]_+ \ge \eta  \right) =0\ .
		 		\end{align*}
		 		Furthermore, since $\ln(1-x) = -x$ for $0<x<1$,
		 		\begin{align*}
		 		\mathds{P} &\Big( \max \big\{ \hat l_j : j \in 
J_{\lceil \nu L_N/4 \rceil} \big\} \le \nu^{-1} \big[ \ln(L_N) + \ln( 4 C_1 / 
3) \big] \Big)  \\
		 		&\le \,  \left( 1 - \dfrac{1}{(4/3) L_N C_1} 
\right)^{2 \lceil \nu L_N / 4\rceil} \le \mathrm{e}^{- 3 \nu/(8 C_1)}
		 		\end{align*}
		 		for all but finitely many $N \in \mathds N$. Hence, altogether one obtains, with $\eta = - \nu^{-1} \ln(3/4)$,
		 		\begin{align*}
		 		 \mathds{P} &\Big( \llargest > \nu^{-1} \left[ \ln(L_N) + \ln \left( C_1 \right) \right] \Big)\\
		 		&\ge \,  \mathds{P} \Big( \llargest > \max \big\{ \hat l_j : j \in J_{\lceil \nu L_N/4 \rceil} \big\} - \eta \Big) \, \\
		 		& \quad + \, \mathds{P} \Big( \max \big\{ \hat 
l_j : j \in J_{\lceil \nu L_N/4 \rceil} \big\} > \nu^{-1} \left[ \ln(L_N) + \ln 
\left( 4C_1/3 \right) \right] \Big)  - 1 \\
		 		&\ge \,  1 - \dfrac{1}{2} \eta'
		 		\end{align*}
		 		for all but finitely many $N \in \mathds N$.
		 	\end{proof}

		 	Recall that $\llargest = l_N^{>,1}$ is the largest and $l^{>,k}_N$, $k \in \mathds N$, is the $k$th largest length of the set $\{l_j= |(x_j(\omega), x_{j+1}(\omega)) \cap \Lambda_N| : j \in \mathds Z\}$. In the same way we define $\tilde{l}^{>,k}_N$ for the scaled lengths.
		\begin{theorem} \label{Theorem E-Gap}
		 For any $0<\eta'<2$ and any $C_3
> 2$ there exists an $\widetilde N=\widetilde N(\eta',C_3) \in \mathds N$ such that for any $N \ge 
\widetilde N$
			\begin{align}
			\mathds{P} \left( \ell^{>}_{N} > \nu^{-1} \ln ( L_N C_1) \, , \, l^{>,\left\lceil 2\nu C_3 / (C_1 \eta') \right\rceil + 1}_N \le \nu^{-1} \left[ \ln(L_N C_1) - \ln(C_3/2) \right] \right) > 1 - \eta'
			\end{align}
			with $C_1= - \nu / [4 \ln(\eta' / 2)]$.
		\end{theorem}
		\begin{proof}
			According to Theorem~\ref{theorem C3} there exists a number $\widetilde N(\eta') \in \mathds N$ such that 
for any $N \ge \widetilde N(\eta')$ one has, $\Omega_1:=\{\omega \in \Omega:\ \ell^{>}_{N} > \nu^{-1} \ln(C_1 L_N) \}$ ,
			\begin{align}
			\mathds{P} (\Omega_1) > 1 - \dfrac{1}{2} \eta'\ .
			\end{align}
			Moreover,
			\begin{align}
			\E \left[ \mathcal N_N^{\text{I}, \omega}(E) \right] \le \mathcal N_{\infty}^{\text{I}}(E)
			\end{align}
			for any $E \ge 0$ and $N \in \mathds N$ \cite[Theorem 5.25]{pastur1992spectra}.
			Here, $\mathcal N_N^{\text{I}, \omega}(E) = L_N^{-1} \left|\left\{ i :  E_N^{i,\omega} \le E \right\}\right|$ is the finite-volume integrated density of states, i.e., the number of eigenvalues of the non-interacting Luttinger--Sy model that are smaller than or equal to $E$ divided by the volume of the system, and
			\begin{align}
			\mathcal N_{\infty}^{\text{I}}(E) = \nu \dfrac{\mathrm{e}^{- \nu \pi E^{-1/2}}}{1 - \mathrm{e}^{- \nu \pi E^{-1/2}}}
			\end{align}
			is the limiting integrated density of states of the non-interacting Luttinger--Sy model, see e.g. \cite[Proposition III.2]{zagrebnov2007bose}. 
			
			Hence, $\widetilde E := \pi^2 \nu^{2} [ \ln(C_1 L_N) - \ln(C_3/2) ]^{-2}$, 
			\begin{align*}
			k \, \mathds{P}\left( \left|\left\{ i :  E_N^{i,\omega} \le \widetilde E \right\}\right| \ge k \right) & \le \sum\limits_{j \geq 1} j \, \mathds{P} \left( \left|\left\{ i :  E_N^{i,\omega} \le \widetilde E \right\}\right| = j \right) \\
			& = \E \left[ \left|\left\{ i :  E_N^{i,\omega} \le \widetilde E \right\}\right|\right] \le \nu\dfrac{C_3}{C_1}
			\end{align*}
			for any $k \in \mathds N$ and all but finitely many $N \in \mathds N$. Setting 

			\begin{align*}
			\Omega_2:=\left\{\omega \in \Omega:  \left|\left\{ i :  E_N^{i,\omega} \le \widetilde E \right\}\right| < \left\lceil \dfrac{2\nu C_3}{C_1 \eta'} \right\rceil\right\}
			\end{align*}
one obtains
			\begin{align}
			\mathds{P}(\Omega_2) \ge 1 - \dfrac{1}{2} \eta'\ .
			\end{align}
			for all but finitely many $N \in \mathds N$. 
			
			Finally,  $\mathds{P}(\Omega_1 \cap \Omega_2) \ge 1 - \eta'$ for all but finitely many $N \in \mathds N$, and for any $\omega \in \Omega_ 1 \cap \Omega_2$ one has
			\begin{align*}
			\ell^{>}_{N}> \nu^{-1} \ln(C_1L_N) 
			\end{align*}
			and
			\begin{align*}
			l_N^{>,\left\lceil 2 \nu C_3/ (C_1 \eta')\right\rceil + 1} \le \nu^{-1} \left[ \ln(C_1L_N) - \ln\left(\dfrac{C_3}{2} \right) \right] \ .
			\end{align*}
			Consequently,  $\llargest - l_N^{>, \left\lceil 2\nu C_3 / (C_1 \eta') \right\rceil + 1} \ge \nu^{-1} \ln(C_3/2)$ for any $\omega \in \Omega_ 1 \cap \Omega_2$. 
		\end{proof}
		\begin{cor}\label{Corollary E-Gap nuN}
			Let $1/N \ll \nu_N \lesssim 1$ be given. Then for any $0 < \eta ' < 2$ and any $C_3 > 2\mathrm{e}^{\nu/\rho}$ there is an $\widetilde N=\widetilde N(\eta^{\prime},C_3)
\in \mathds N$ such that for any $N \ge \widetilde N$ we have
			\begin{align}\label{C.39}
			\mathds{P} \left(\llargesttilde > \nu_N^{-1}\ln(C_1L_{\lfloor s_N N \rfloor})\, , \, \llargesttilde - \tilde l_N^{>,\lceil 2\nu C_3 / (C_1 \eta') \rceil + 1} > \nu_N^{-1} \ln(C_3/(2\mathrm{e}^{\nu/\rho})) \right) > 1 - \eta'
			\end{align}
			with $C_1 := - \nu / [4 \ln(\eta' / 2)]$.
		\end{cor}
		\begin{proof}
			%

			%
			%

			With Theorem~\ref{Theorem E-Gap} there exists an $\widehat N=\widehat N(\eta^{\prime}) \in \mathds{N}$ such that for any $N \ge \widehat{N}$ one has
			%
	
			%
			\begin{equation}\label{GLEICHUNGXXX}\begin{split}
			 &\mathds{P} \left(s_N^{-1}\ell^{>}_{\lfloor s_N N \rfloor} > \nu_N^{-1} \ln(C_1L_{\lfloor s_N N \rfloor})\, , \, s_{N}^{-1}l^{>,\left\lceil 2\nu C_3 / (C_1 \eta') \right\rceil + 1}_{\lfloor s_N N \rfloor} < \nu^{-1}_{N} \left[ \ln(C_1L_{\lfloor s_N N \rfloor}) - \ln(C_3/2) \right] \right) \\ & \qquad > 1 - \eta' \ .
			\end{split}
			\end{equation}
			As in the proof of Theorem~\ref{theorem lower bound largest interval} we use $\llargesttilde \ge s_N^{-1} \ell_{\lfloor s_N N \rfloor}^>$ as well as
			%
			%
		$\tilde l_{N}^{>, \lceil 2\nu C_3 / (C_1 \eta')\rceil+1} \leq s_N^{-1} l_{\lfloor s_N N \rfloor +1}^{>, \lceil 2\nu C_3 / (C_1 \eta')\rceil+1} $ to obtain
		\begin{align*}
		\llargesttilde - \tilde l_N^{>,\lceil 2\nu C_3 / (C_1 \eta') \rceil + 1} &\geq s_{N}^{-1}\left(l^{>}_{\lfloor s_N N \rfloor}-l^{>,\left\lceil 2\nu C_3 / (C_1 \eta') \right\rceil + 1}_{\lfloor s_N N \rfloor+1}\right) \, \\
		&\geq s_{N}^{-1}\left(l^{>}_{\lfloor s_N N \rfloor}-\left(l^{>,\left\lceil 2\nu C_3 / (C_1 \eta') \right\rceil + 1}_{\lfloor s_N N \rfloor}+\rho^{-1}\right)\right)\ .
		\end{align*}
		Now, using the two inequalities appearing in~\eqref{GLEICHUNGXXX} we conclude
		\begin{align*}
		\llargesttilde - \tilde l_N^{>,\lceil 2\nu C_3 / (C_1 \eta') \rceil + 1} &\geq \nu_{N}^{-1}\ln\left(\frac{C_3}{2\mathrm{e}^{\nu/\rho}}\right)
		\end{align*}
		from which the statement readily follows.
		\end{proof}
By slightly changing the proof we could also allow for $C_3>6$ instead of $C_3 > 2\mathrm{e}^{\nu/\rho}$. This would replace $\ln(C_3/2\mathrm{e}^{\nu/\rho})$ in \eqref{C.39} by $\ln(C_3/6)$.

	{\small
		\bibliographystyle{amsalpha}
		\bibliography{Literature}}

\end{document}